\documentclass[%
 aip,
 amsmath,amssymb,
 reprint,%
]{revtex4-1}

\usepackage{graphicx}
\usepackage{dcolumn}
\usepackage{bm}
\usepackage[utf8]{inputenc}
\usepackage[T1]{fontenc}
\usepackage{mathptmx}
\usepackage{amsmath,amssymb,amsthm,color,dsfont,mathtools,fullpage,setspace}
\usepackage{bibentry}
\usepackage{nicefrac}
\usepackage{siunitx}
\usepackage{xcolor}
\usepackage{soul}
\usepackage{xr-hyper}

\usepackage{hyperref}
\usepackage{tabularx,booktabs}
\newcolumntype{Y}{>{\centering\arraybackslash}X}
\newcommand{\<}{\left \langle}
\renewcommand{\>}{\right \rangle}
\newcommand{\1}{\mathds{1}}

\DeclareMathOperator{\cov}{cov}
\DeclareMathOperator{\var}{var}

\newcommand{\R}{\mathbb{R}}
\newcommand{\Obs}{\mathcal{G}}
\newcommand{\MBARCov}{\mathcal{A}}
\newcommand{\grad}{\nabla}
\newcommand{\RFN}{F} 
\newcommand{\VFS}{y}  
\newcommand{\VAS}{z}  

\newcommand{\qmix}{\pi_{\mathrm{mix}}}
\newcommand{\E}{\mathbf{E}}
\newcommand{\dw}{\overline{dw}}

\theoremstyle{plain}
\newtheorem{theorem}{Theorem}[section]
\newtheorem{assumption}[theorem]{Assumption}
\newtheorem{lemma}[theorem]{Lemma}

\begin{document}

\preprint{AIP/123-QED}

\title{Understanding the Sources of Error in MBAR through Asymptotic Analysis}


\author{Xiang Sherry Li}
\affiliation{
Department of Chemistry and James Franck Institute, University of Chicago, Chicago, Illinois 60637, USA
}
\author{Brian Van Koten}
\affiliation{Department of Mathematics and Statistics, University of Massachusetts, Amherst, Massachusetts 01003, USA}
\author{Aaron R. Dinner}
\affiliation{
Department of Chemistry and James Franck Institute, University of Chicago, Chicago, Illinois 60637, USA
}
\author{Erik H. Thiede}
\affiliation{Center for Computational Mathematics, Flatiron Institute, New York, New York 10010, USA}
\email{ehthiede@flatironinstitute.org}

\begin{abstract}
    Multiple sampling strategies commonly used in molecular dynamics, such as umbrella sampling and alchemical free energy methods, involve sampling from multiple thermodynamic states. 
    Commonly, the data are then recombined to construct estimates of free energies and ensemble averages using the Multistate Bennett Acceptance Ratio (MBAR) formalism.
    However, the error of the MBAR estimator is not well-understood: previous error analysis of MBAR assumed independent samples and did not permit attributing contributions to the total error to individual thermodynamic states. 
    In this work, we derive a novel central limit theorem for MBAR estimates. 
    This central limit theorem yields an error estimator which can be decomposed into contributions from the individual Markov chains used to sample the states.
    We demonstrate the error estimator for an umbrella sampling calculation of the alanine dipeptide in two dimensions and an alchemical calculation of the hydration free energy of methane. 
    In both cases, the states' individual contributions to the error provide insight into the sources of error of the simulations. 
    Our numerical results demonstrate that the time required for the Markov chain to decorrelate in individual thermodynamic states contributes considerably to the total MBAR error.
    Moreover, they indicate that it may be possible to use the contributions to tune the sampling and improve the accuracy of MBAR calculations.
\end{abstract}

\maketitle

\section{Introduction}

Molecular dynamics simulations provide a means to compute key quantities in statistical mechanics, typically in the form of ensemble averages of certain observables. 
In principle, it is possible to estimate ensemble averages by running a long, unbiased simulation of a system and averaging over the resulting trajectory. 
However, in practice this can be inefficient. 
For instance, the determination of the relative free energy is of particular interest since the free energy is the fundamental quantity that determines the spontaneity of thermodynamic processes and the accessibility of thermodynamic states.
However, the ensemble averages required for estimating free energies are notoriously slow to converge and typically require simulations to explore multiple regions of a system's configurational space\cite{chipot2007free,lelievre2010free}.
 
A common strategy for addressing this problem is to sample from multiple thermodynamic states.
If the states are chosen well, ensemble averages will converge much more quickly in these new states\cite{dinner2020stratification}.
This approach is the basis of multiple simulation strategies such as umbrella sampling, parallel tempering, and alchemical free energy simulation.
Umbrella sampling enforces sampling of the conformational changes of biomolecular systems by running a series of independent simulations, each biased to sample a different region of a physical system's phase space\cite{torrie1977nonphysical,Pangali1979Monte}. 
Parallel tempering simulates several replicas of a system with the same Hamiltonian but at different temperatures to allow the system to cross over energetic barriers\cite{swendsen1986replica,Geyer1991paralleltemp}.
Alchemical free energy calculations of molecular systems estimate the free energy difference between two molecular states by interpolating between their Hamiltonians\cite{tembre1984ligand,FrenkelSmit2002,CHODERA2011150}.
In all of these cases, data sampled from multiple thermodynamic states must be combined to obtain the averages of interest.

A popular algorithm for doing this is the Multistate Bennett Acceptance Ratio (MBAR)\cite{shirts2008statistically}, originally derived in~\onlinecite{vardi1985empirical}.
For statistically independent samples, MBAR combines the data across states in a statistically optimal manner.
Although molecular dynamics simulations are correlated, MBAR nevertheless achieves good results in practice.
However, solving for the MBAR estimate involves solving a nonlinear fixed point problem
which complicates error analysis.  
Previous attempts to construct error estimates have explicitly assumed that samples are statistically independent\cite{kong2003theory,shirts2008statistically}, which is typically not true for molecular dynamics data.

In the present paper, we build on~\onlinecite{thiede2016eigenvector,dinner2020stratification} to derive a formal expression for the asymptotic variance of the MBAR estimator that explicitly accounts for correlation in sampled data. Moreover, our estimator can be decomposed into error contributions from individual states. This can potentially give practitioners insight into how sampling in individual thermodynamic states affects global error, and may lead to adaptive sampling strategies to accelerate convergence. 

\section{Background on Monte Carlo and Asymptotic Error}\label{sec:background}
Many fundamental quantities in statistical mechanics take the form of high-dimensional integrals.
Physical theories often require estimating averages over a physical system's Boltzmann distribution of the form
\begin{equation}
    \< g\>  = \frac{\int g(x) e^{- H(x) / k_B T}dx}{\int e^{- H(x) / k_B T}dx },
    \label{eq:ensemble_average}
\end{equation}
where $H$ is the system's Hamiltonian and $x$ is its configuration in $\R^n$.
Alternatively, they may require estimates of the free energy difference between two regions of phase space:
\begin{align}
    \Delta G_{A \to B} 
    =& -k_B T \log  \frac{\int \1_B(x) e^{-H(x) / k_B T} dx}{\int \1_A(x) e^{-H(x) / k_B T} dx}
    \label{eq:free_energy_sets_integral} 
    \\
    =& -k_B T \log  \frac{\<\1_B \>}{\<\1_A\>},
    \label{eq:free_energy_sets}
\end{align}
where $\1_D(x)$ is $1$ if the configuration $x$ is in a region labeled $D$ and $0$ otherwise.
We might also wish to estimate the free energy difference between two Hamiltonians, given by
\begin{equation}
    \Delta G_{\alpha\to\beta} =-k_B T \log  \frac{\int  e^{-H_\beta(x) / k_B T} dx}{\int e^{-H_\alpha(x) / k_B T} dx}.
    \label{eq:free_energy_hams}
\end{equation}
For most systems, these integrals are too complex to be evaluated analytically, and the dimension of $x$ is too high to use quadrature.
Instead, they are typically treated by Monte Carlo methods.

\subsection{Markov chain Monte Carlo}

Assume that we are given a probability distribution with an unnormalized density $q$ over the Lebesgue measure on $\R^n$.
For instance, in statistical mechanics $q$ is typically the Boltzmann factor:
\begin{equation}
    q(x) = e^{-H(x) / k_B T}.
\end{equation}
We can then write the average of a function $g : \R^n \to \R$ over the distribution as 
\begin{equation}
    \< g\> =\frac{\int g(x) q(x) dx}{c}, \qquad c=\int q(x) dx.
\end{equation}
In statistical mechanics, we refer to these averages as \emph{ensemble averages}.
Particular choices of $g$ allow us to rewrite key quantities in statistical mechanics as ensemble averages.
Substituting the Boltzman factor for $q$ directly recovers~\eqref{eq:ensemble_average},
and the free energy differences between regions of phase space in~\eqref{eq:free_energy_sets} is merely the ratio between two averages.
Similarly, to estimate the free energy difference between two Hamiltonians, we can set $q$ to $\exp\left( -H_\alpha / k_B T\right)$ and rewrite~\eqref{eq:free_energy_hams} as  
\begin{equation}
    e^{-\Delta G_{\alpha\to\beta}/k_B T} = 
    \< e^{-\left(H_\beta(x) - H_\alpha(x)\right) / k_B T} \>.
    \label{eq:zwanzig_relation}
\end{equation}

Monte Carlo methods approximate ensemble averages
by drawing a sequence of $N$ samples $\left\{X_t\right\}$ from the probability distribution and averaging over them.
If the sampling procedure is chosen appropriately, we expect sample averages to converge to the true (ensemble) average over $p$.
\begin{equation}
    \bar{g} = \frac{1}{N} \sum_{t=0}^{N-1} g\left( X_t \right) \xrightarrow{a.s.} \<g\>.
    \label{eq:iid_lln}
\end{equation}
Here $a.s.$ denotes almost sure convergence, a strong form of probabilistic convergence.
(Specifically, the probability of $\bar{g}$ not converging to $\<g\>$ is zero.)
If the samples are statistically independent we say that our samples are \emph{independent and identically distributed} (IID), and~\eqref{eq:iid_lln} is guaranteed to hold by the Law of Large Numbers\cite{lelievre2010free}.   
However, in practice it is often impossible to generate IID samples, and we must instead generate samples by running a Markov chain that has $p$ as its stationary distribution: a practice known as \emph{Markov chain Monte Carlo} (MCMC).
Then,~\eqref{eq:iid_lln} still holds if the Markov chain is ergodic~\cite{lelievre2010free}.

\subsection{Asymptotic Variance of Monte Carlo Estimates}\label{ssec:avar_mcmc}

While~\eqref{eq:iid_lln} guarantees that the error goes to zero as the number of samples increases, it says nothing about how quickly this happens.
A common method to quantify how the sampling error decreases increasing sample size is to use a \emph{Central Limit Theorem} (CLT): 
a theorem showing that a sequence of random variables converges to a known normal distribution\cite{lelievre2010free}.
Specifically, assume that we wish to evaluate the ensemble average of several functions, each denoted $g_i$.
Concatenating our sample means and ensemble averages into vectors we denote as $\bar{\mathbf{g}}$ and $\<\mathbf{g}\>$ respectively,
we can often show that the error between the two converges to a multivariate normal (Gaussian) distribution when appropriately scaled:
\begin{equation}
    \sqrt{N} \left(\bar{\mathbf{g}} -\<\mathbf{g}\> \right) \xrightarrow{d} \mathbf{N}\left(0, \Sigma \right).
    \label{eq:multivariate_mcmc_clt}
\end{equation}
Here $\mathbf{N}$ denotes a normal (Gaussian) random variable with mean vector $\mu$ and covariance matrix $\Sigma$, known as the \emph{asymptotic covariance}.
The symbol $\xrightarrow{d}$ denotes convergence in distribution (a weaker form of probabilistic convergence than almost sure convergence).
For IID samples,~\eqref{eq:multivariate_mcmc_clt} holds for all functions with finite variance and $\Sigma_{ij}$ is simply the covariance between $g_i$ and $g_j$ over $p$.
If samples are instead gathered from a Markov chain $X_t$, proving a CLT requires more technical conditions on the nature of the Markov Chain and $\mathbf{g}$ \cite{chan1993central,tierney1994markov,geyer1994discussion,jones2004markov}.
However, for most commonly used Markov chains and most reasonable applications, we can expect~\eqref{eq:multivariate_mcmc_clt} to hold.
In this case, if the Markov chain is time-homogeneous (i.e., the rule for updating the Markov chain is the same at all times) with ergodic distribution $p$, the asymptotic covariance matrix is given by
\begin{equation}\label{eq: first integrated autocovariance formula}
    \Sigma_{ij}=\cov\{g_i(X_t),g_j(X_t)\}+2\sum_{k=1}^\infty \cov\{g_i(X_t),g_j(X_{t+k})\},
\end{equation}
where in this formula (and only this formula) we assume that the chain is stationary with $X_t$ distributed according to $p$ for all $t$. 

The CLT and the asymptotic covariance help diagnose the error and convergence of a Markov chain Monte Carlo simulation.
For example, under mild technical conditions (specifically geometric ergodicity and bounded $g$), $\Sigma_{ii}/N$ is asymptotic to $\var\{\bar g_i \}$. 
Consequently, we can treat $\sqrt{\Sigma_{ii}/N}$ as a rough estimate for the error in using $\bar g_i$ to estimate of $\langle g_i \rangle$.
The sampling efficiency of the Markov chain relative to IID sampling from $p$ can be quantified by the \emph{autocorrelation time}:
\begin{equation}
    \tau_{g_i} = \frac{\Sigma_{ii}}{\var \left\{g_i \right\}}.
    \label{eq:defn_actime}
\end{equation}
Equivalently,
\begin{equation*}
    \Sigma_{ii} = \tau_{g_i} \var \left\{g_i \right\}.
\end{equation*}
Since $\var \left\{ g_i \right\}/N$ is the variance for IID sampling, we can interpret the autocorrelation time as how many MCMC samples are required to achieve the same reduction in error as a single IID sample\cite{lelievre2010free, FrenkelSmit2002}.

\section{The MBAR Equations}\label{sec:MBAR}
In the previous section, we considered sampling from a single distribution.
However, we may often have samples collected from multiple, related probability distributions.
For concreteness, assume we have $L$ probability distributions, each with an unnormalized probability density $q_i$. We refer to these distributions as \emph{states}.
The ensemble average of an observable $g(x)$ in each state is given by
\begin{equation}
    \label{eq:ensemble_average_over_states}
    \< g\> _i = \frac{\int g(x) q_i(x)dx}{c_i},\quad c_i = \int q_i(x) dx.
\end{equation}
Here the constant $c_i$ is the normalization constant for $q_i$.
If $q_i$ is a Boltzmann distribution, then $c_i$ is the corresponding partition function.
Next, we assume that for every state we have collected a set of $N_i$ samples, denoted $\{X^i_t\}$.
We can then approximate $\< g\>_i $ by the sample average
\begin{equation*}
    \< g \>_i \approx \frac{1}{N_i}\sum_{t=1}^{N_i}g(X_t^i)
    \label{eq:empirical_estimate_over_states}
\end{equation*}
However, if the states have shared regions with high probability,
we can construct improved estimates of~\eqref{eq:ensemble_average_over_states} by using data from all of the states, not just state $i$.
This is the aim of the MBAR algorithm\cite{vardi1985empirical,shirts2008statistically}.
Following the treatment in Ref.\ \onlinecite{geyer1994estimating},
we observe that we can view the union of the samples from the states as samples from a combined distribution
known as a \emph{mixture distribution}.
Let
\begin{equation*}
  N = \sum_{i=1}^L N_i
\end{equation*}
be the total sample size, and let
\begin{equation*}
  \kappa_i = \frac{N_i}{N}
\end{equation*}
be the fraction of sample points collected in state $i$.
To simplify the presentation, we will assume that $\kappa_i$ is constant and always greater than zero.
(A version of our main result that relaxes this assumption is given in the supplement.)
We define the mixture distribution to be
\begin{equation}
    \pi_{\text{mix}}(x)=\sum_{i=1}^L\kappa_i  q_i(x)  / c_i.
\end{equation}
We can then write
\begin{align}
    \<g\>_i
        &=\int\frac{ g(x) q_i(x) / c_i }{\pi_\text{mix}(x)}\pi_\text{mix}(x)dx \nonumber \\
        &=\int\frac{g(x)q_i(x) / c_i }{\sum_{k=1}^L\kappa_k q_k(x) / c_k }\sum_{j=1}^L\kappa_j q_j(x) / c_j dx\\
        &= \sum_{j=1}^L\kappa_j\<\frac{g q_i / c_i }{\sum_{k=1}^L\kappa_k q_k / c_k }\>_j.
    \label{eq:intermediate_reweighting}
\end{align}
In general, the normalization constants for the states are not known.  
We therefore rewrite~\eqref{eq:intermediate_reweighting} in terms of the states' (unitless) relative free energies, which we denote $f_i$.
We arbitrarily set the average free energy to be zero, so 
\begin{equation}
    \frac{1}{L}\sum_{i=1}^L f_i = 0,
    \label{eq:normalization_of_zs}
\end{equation}
and therefore the free energies are defined by
\begin{equation}
    f_i = - \log c_i + \frac{1}{L}\sum_{j=1}^L \log c_j.
\end{equation}
Dividing both the numerator and denominator of~\eqref{eq:intermediate_reweighting} by $\exp(-(1/L) \sum_{i=1}^L f_i)$,
after a few manipulations we have 
\begin{align}
    \<g\>_i = \sum_{j=1}^L\kappa_j\<\frac{g q_i e^{f_i} }{\sum_{k=1}^L\kappa_k q_k e^{f_k} }\>_j.
    \label{eq:mbar_average}
\end{align}
This equation can be used to estimate $\<g\>_i$ if we are given an estimate of the free energies, $\bar{f}$.
Replacing each ensemble average on the right-hand side with a Monte Carlo estimate
we have
\begin{equation}
    \bar{g}_i = \sum_{j=1}^L\frac{\kappa_j}{N_j} \sum_{t=1}^{N_j}\frac{ g(X_t^j) q_i(X_t^j) e^{\bar{f}_i}
    }{\sum_{k=1}^L\kappa_k q_k(X_t^j) e^{\bar{f}_k} }.
    \label{eq:estimated_mbar_avg}
\end{equation}
This estimator uses data from every state, not just state $i$.
Moreover, we can also use~\eqref{eq:mbar_average}
to estimate the free energies themselves.
Since the ensemble average of the function $g(x)=1$ is always $1$,
\begin{align}
    1 =& \sum_{j=1}^L\kappa_j\<\frac{q_i e^{f_i} }{\sum_{k=1}^L\kappa_k q_k e^{f_k}}\>_j \\
    \implies  f_i =& - \log \sum_{j=1}^L \kappa_j\<\frac{q_i }{\sum_{k=1}^L\kappa_k q_k e^{f_k}}\>_j.
    \label{eq:z_iteration}
\end{align}
One can thus estimate the free energy by defining $\bar{f}$ to be the solution to 
\begin{equation}
    \bar{f}_i = - \log \left(\sum_{j=1}^L\frac{\kappa_j}{N_j} \sum_{t=1}^{N_j}\frac{q_i(X_t^j)}{\sum_{k=1}^L\kappa_k q_k(X_t^j)e^{\bar{f}_k}} \right).
    \label{eq:estimated_z_iteration}
\end{equation}
Not only can this equation can be solved using standard root-finding methods such as Newton-Raphson and gradient descent\cite{shirts2008statistically},
but there exist rapid algorithms for solving it through a succession of estimation tasks\cite{meng1996simulating,thiede2016eigenvector,dinner2020stratification}.
Equations~\eqref{eq:estimated_mbar_avg} and~\eqref{eq:estimated_z_iteration} are the MBAR estimates of the ensemble average and the free energies, respectively\cite{shirts2008statistically}.
With sufficient overlap of the samples from different states, \eqref{eq:z_iteration} uniquely determines $f$.
Specifically, if the matrix $M_{ij} = \<q_i\>_j$ is irreducible, then by Theorem~1 in Ref.~\onlinecite{geyer1994estimating} or Proposition~1.1 in Ref.~\onlinecite{gill_large_1988} the $f_i$
are uniquely specified by~\eqref{eq:z_iteration}.
An analogous statement holds for $\bar{f}$.
When $M$ is irreducible, Theorem~1.1 in Ref.~\onlinecite{gill_large_1988} implies that equation~\eqref{eq:estimated_z_iteration} almost surely has a unique solution $\bar f$ when the total sample size $N$ is sufficiently large. Moreover, the estimates of the free energies and the ensemble averages converge to the true values as $N$ increases. To be precise, $\bar f_i \xrightarrow{a.s.} f_i$ and $\bar g_i \xrightarrow{a.s.} \langle g \rangle_i$ by Theorem~1 in Ref.~\onlinecite{geyer1994estimating}.

\subsection{Estimating Chemical Quantities using MBAR}

Specific manipulations of state free energies and ensemble averages  allow us to reconstruct quantities of interest in a broad range of contexts.
Here, we discuss the analysis of data from three common algorithms: parallel tempering, alchemical free energy simulations, and umbrella sampling.

In parallel tempering, we seek to estimate ensemble averages for a system with unnormalized probability density
\begin{equation}
    e^{-H(x) / k_B T}.
    \label{eq:basic_Boltzmann}
\end{equation}
However, this density may be highly multimodal, making the probability distribution difficult to sample.
Parallel tempering addresses this by running multiple copies of the system with the same Hamiltonian but different temperatures\cite{swendsen1986replica,Geyer1991paralleltemp}.  We write their distributions as
\begin{equation}
    q_i(x) = e^{-H(x)/k_B (T + \delta T_i)}
  \end{equation}
One copy, here arbitrarily chosen to have index 1, is set to be at the original temperature (i.e., $\delta T_1 = 0$)
and all other copies have $\delta T_i \neq 0$.
The copies then periodically swap molecular configurations via Monte Carlo moves on the space of copies.
In principle, one can estimate averages over~\eqref{eq:basic_Boltzmann} by averaging over all configurations collected from $q_1$.
However, using the MBAR estimator~\eqref{eq:estimated_mbar_avg} allows one to use data from all states, giving a more accurate answer.

In alchemical free energy simulations, we seek to estimate the free energy difference between two Hamiltonians as in~\eqref{eq:free_energy_hams}\cite{tembre1984ligand,FrenkelSmit2002,CHODERA2011150}.
However, rather than sampling only the state with $H_\alpha$, we sample a set of $L$ states that interpolate between $H_\alpha$ and $H_\beta$.  A simple choice would be to set
\begin{equation}
-k_B T \log q_i = H_\alpha + 
\lambda\left(\frac{i-1}{L-1} \right) (H_\beta - H_\alpha)
    \label{eq:alchemical_interpolation}
\end{equation}
where $\lambda:\left[0,1 \right] \to \left[0,1 \right]$ is a monotonic function such that $\lambda(0)=0$ and $\lambda(1)=1$,
although in practice, more complex interpolations are often required \cite{simonson1993free,beutler1994avoiding,steinbrecher2007nonlinear}.
With this set of state definitions, the (unitless)
free energy difference between the two Hamiltonians is simply
the difference between the free energies of the first and last states.
\begin{equation}
    - \log \frac{\int e^{-H_\beta(x)/k_B T} dx}{\int e^{-H_\alpha(x)/k_B T} dx} = f_L - f_1
    \label{eq:mbar_state_fe_diff}
\end{equation}
Consequently, we can solve~\eqref{eq:estimated_z_iteration} and estimate the free energy difference as $\bar{f}_L - \bar{f}_1$.

In umbrella sampling\cite{torrie1977nonphysical,Pangali1979Monte} we construct a collection of states
\begin{equation}
      q_i(x) = \psi_i(x) q(x)
\end{equation}
by multiplying a density $q$ with a biasing function $\psi_i$.
We then aim to estimate averages of observables over $q$, such as those in~\eqref{eq:ensemble_average} and~\eqref{eq:free_energy_sets}.
To estimate these averages over $q$ using MBAR by steps similar to those used to derive~\eqref{eq:intermediate_reweighting}, we write
\begin{align}
  \frac{\int g(x) q(x) dx}{\int q(x) dx} 
  =& \frac{\int g(x) q(x) (\pi_{\text{mix}} (x) / \pi_{\text{mix}}(x)) dx}{\int q(x) (\pi_{\text{mix}}(x) / \pi_{\text{mix}}(x))dx} \nonumber \\
  =& \frac{\int g(x) q(x) \frac{\sum_j \kappa_j q_j(x) e^{f_j} }{\sum_l \kappa_l q_l(x) e^{f_l}} \, dx}{\int q(x) \frac{\sum_k \kappa_k q_k(x) e^{f_k} }{\sum_m \kappa_m q_m(x) e^{f_m}} \, dx} \nonumber \\
  =& \frac{\sum_{j=1}^L\kappa_j\<{g q}/{(\sum_{l=1}^L\kappa_l q_l e^{f_l})}\>_j }
     {\sum_{k=1}^L\kappa_k\<{q}/{(\sum_{m=1}^L\kappa_m q_m e^{f_m})}\>_k} \nonumber \\
  =& \frac{\sum_{j=1}^L\kappa_j\<{g }/{(\sum_{l=1}^L\kappa_l \psi_l e^{f_l})}\>_j }
     {\sum_{k=1}^L\kappa_k\<{1}/{(\sum_{m=1}^L\kappa_m \psi_m e^{f_m})}\>_k}.
     \label{eq:us_average}
\end{align}
We can also use umbrella sampling to estimate the difference in free energy between two states.
Comparing to~\eqref{eq:free_energy_sets} and setting $q$ to be the Boltzmann factor, 
we have
\begin{align}
    \Delta G_{A \to B} 
        =& -k_B T \log  \frac{\int \1_B(x) e^{-H(x) / k_B T} dx}{\int \1_A(x) e^{-H(x) / k_B T} dx} \nonumber \\
        =& -k_B T \log  
            \frac{
                \sum_{j=1}^L\kappa_j\<{ \1_A q }/{(\sum_{l=1}^L\kappa_l q_l e^{f_l})}\>_j 
            }{
                \sum_{k=1}^L\kappa_k\<{\1_B q}/{(\sum_{m=1}^L\kappa_m q_m e^{f_m})}\>_k
            } \nonumber \\
        =& -k_B T \log  
            \frac{
                \sum_{j=1}^L\kappa_j\<{ \1_A }/{(\sum_{l=1}^L\kappa_l \psi_l e^{f_l})}\>_j 
            }{
                \sum_{k=1}^L\kappa_k\<{\1_B}/{(\sum_{m=1}^L\kappa_m \psi_m e^{f_m})}\>_k
            },
\end{align}
by steps similar to those for~\eqref{eq:us_average}.

These examples show how MBAR can be used to efficiently construct estimates from algorithms that collect data in multiple states.
Indeed, when IID samples are collected from each state, then MBAR gives the maximum likelihood estimate \cite{vardi1985empirical} and achieves the best possible mean-squared error in the large-sample limit \cite{shirts2008statistically}.
MBAR does not give the maximum likelihood estimate for correlated samples, but nevertheless it has been empirically observed to give good results.

However, it is not obvious how to estimate the uncertainty in MBAR averages.
In previous work, Kong {\it et al.} constructed an estimator for the asymptotic covariance using the Cramer-Rao lower bound of the  variance \cite{kong2003theory}.
When samples are uncorrelated, MBAR achieves this lower bound. However, when samples are correlated, this estimator underestimates the asymptotic error.  
Moreover, the estimator only gives the total asymptotic error and not the contributions from individual thermodynamic states. 
This makes it difficult to determine how the parameters of individual states contribute to the accuracy of the total simulation.
An alternate approach would be to attempt to construct a CLT for MBAR estimates.
As discussed in Subsection~\ref{ssec:avar_mcmc}, CLTs are able to capture the effect of the dynamics on sampling error.
Moreover, previous work on other algorithms for recombining data from multiple states\cite{thiede2016eigenvector} has shown that CLTs can be used to connect the sampling of individual states to the total error of the estimate.
In this work, we establish a CLT for the MBAR equations and show that the resulting error estimates gives detailed insight into how the parameters of multistate simulations contribute to the total error.

\section{Asymptotic Variance for the MBAR equations} 
Here, we give an expression for the asymptotic variance of MBAR estimates of observables and normalization constants.
Our approach builds upon the work of Geyer\cite{geyer1994estimating}.
Our contribution is essentially to fill in missing details and to correct errors.
Most importantly, the formula for the asymptotic variance of observable averages $\bar{g}$ is not correct in Ref.\ \onlinecite{geyer1994estimating}.

\subsection{CLTs for the raw output of MBAR}\label{ssec:delta_method}

MBAR estimates of observables require calculating the values of $\bar f$ as well as one or more empirical averages of the form
\begin{equation}
    \bar{\omega} =\sum_{j=1}^L \frac{\kappa_j}{N_j} \sum_{t=1}^{N_j} \frac{w(X_t^j)}{\sum_{k=1}^L \kappa_k q_k (X_t^j) e^{\bar{f}_k} }
    \label{eq:generic_component_average}
\end{equation}
for some function $w:\R^{n}\to\R$.
For instance, in~\eqref{eq:estimated_mbar_avg} we set $w = q_i g$ and subsequently multiply by $e^{\bar{f}_i}$.
The presence of $\bar{f}$ in~\eqref{eq:generic_component_average} means that the errors in our observable estimates and in our estimates of the state free energies are correlated.
Consequently, we must consider the asymptotic covariance of the free energies and our observables jointly.

To do so, we rewrite~\eqref{eq:estimated_z_iteration} and~\eqref{eq:generic_component_average} as a single root finding problem.
We concatenate the vector of estimated free energies and empirical averages into a single vector
\begin{equation}
    \bar{v} = \left(\bar{f}_1, \ldots, \bar{f}_L, \bar{\omega}_1, \ldots, \bar{\omega}_M  \right).
    \label{eq:concatenated_sample_means}
\end{equation}
The vector $\bar{v}$ is the root of the function $\bar{\RFN}:\R^{L+M}\to\R^{L+M}$,
where if $i\leq L$ 
\begin{equation}
    \bar{\RFN}_i(y)  = 
            \kappa_i - \sum_{j=1}^L\kappa_j \frac{1}{N_j} \sum_{t=1}^{N_j} \frac{\kappa_i q_i(X_t^j) e^{y_i}}{\sum_{k=1}^L\kappa_k q_k(X_t^j)e^{y_k}}
    \label{eq:root_fxn_states}
\end{equation}
and if  $i > L$ then
\begin{equation}
    \bar{\RFN}_i(y)  = 
            y_i - \sum_{j=1}^L\kappa_j \frac{1}{N_j} \sum_{t=1}^{N_j} \frac{w_{i-L}(X_t^j)}{\sum_{k=1}^L \kappa_k q_k (X_t^j) e^{y_k} }.
    \label{eq:root_fxn_avgs}
\end{equation}
Writing the MBAR estimates as the roots of $\bar{\RFN}$ suggests a strategy for proving a CLT.  For any fixed $y$, each element in $\bar{\RFN}(y)$ is a sum of sample averages over our states.
It is therefore reasonable to assume the existence of a CLT for each of the sample averages.
If we can then convert a CLT for each of the collection of averages into a CLT for the \emph{roots} of $\bar{\RFN}$, then we have proven a CLT for MBAR estimates.
Indeed, this is precisely the strategy we pursue.
A full proof of the CLT is given in Section~I of the supplement.
Here, we merely state introduce the key quantities necessary to state our results and state our result.

We first discuss the asymptotic covariance structure of each of the averages in~\eqref{eq:root_fxn_avgs} and~\eqref{eq:root_fxn_avgs}.
For convenience, we define
\begin{equation}
    \xi_i(x,y) = 
        \begin{cases}
            \frac{\kappa_i q_i(x) e^{y_i}}{\sum_{k=1}^L\kappa_k q_k(x)e^{y_k}} 
                \: \text{if} \; i \leq L 
                \\
            \frac{w_{i-L}(x)}{\sum_{k=1}^L \kappa_k q_k (x) e^{y_k} }
                \: \text{if} \; i > L 
        \end{cases}.
\end{equation}
We can then write $\bar{\RFN}_i(y)$ using a $\kappa_j$-weighted sum of ergodic averages of the form
\begin{equation*}
  \bar \xi^j_i(y) = \frac{1}{N_j} \sum_{t=1}^{N_j} \xi_i(X^j_t,y). 
\end{equation*}
In the limit as $N \rightarrow \infty$, $\bar \xi^j_i(y)$ converges to
\begin{equation*}
  \xi^j_i(y) =  
        \begin{cases}
             \< \frac{\kappa_i q_i  e^{y_i}}{\sum_{k=1}^L\kappa_k q_k e^{y_k}} \>_j
            \quad \text{if} \quad i \leq L \\
            \< \frac{w_{i-L}}{\sum_{k=1}^L \kappa_k q_k  e^{y_k} } \>_j
            \quad \text{if} \quad i > L,
        \end{cases}
\end{equation*}
and $\bar \RFN$ converges to 
\begin{equation}
    \RFN_i(y) = 
        \begin{cases}
            \kappa_i - \sum_{j=1}^L \kappa_j  \< \frac{\kappa_i q_i  e^{y_i}}{\sum_{k=1}^L\kappa_k q_k e^{y_k}} \>_j
            \quad \text{if} \quad i \leq L \\
            y_i - \sum_{j=1}^L \kappa_j \< \frac{w_{i-L}}{\sum_{k=1}^L \kappa_k q_k  e^{y_k} } \>_j
            \quad \text{if} \quad i > L.
        \end{cases}
\end{equation}

We assume that a central limit theorem holds for the ergodic averages $\bar \xi^j_i(y)$.
To be precise, we assume that for any fixed $y$,
\begin{equation}
    \sqrt{N}
    \left(\bar \xi(y)  - \xi(y) \right)
    \xrightarrow{d}
    \mathbf{N}(0, \Xi(y)).
    \label{eq:sample_average_clt}
\end{equation}
Here,
\begin{align*}
  \bar \xi(y) =  (\bar \xi^1_1(y), \dots, \bar \xi^1_{L+M}(y), 
                  \bar \xi^2_1(y), \dots, \bar \xi^2_{L+M}(y), \dots,  \bar \xi^L_{L+M}(y)) 
\end{align*}
is the vector of all ergodic averages, and
$\xi(y)$ is the corresponding vector of limiting values of those averages.
The covariance matrix $\Xi(y)$ can be written in block form as
\begin{equation}
    \def\arraystretch{1.5}
    \Xi(y) =
        \left[
        \begin{array}{c | c | c | c}
            \Xi^{11} (y) & \Xi^{11} (y) & \cdots     & \Xi^{1L}(y)  \\ \hline 
            \Xi^{21} (y) & \Xi^{21} (y) & \cdots     & \Xi^{2L}(y)  \\ \hline 
            \vdots       & \vdots       & \ddots     & \vdots       \\ \hline 
            \Xi^{L1} (y) & \Xi^{L1} (y) & \cdots     & \Xi^{LL}(y)  
        \end{array}
    \right]
\end{equation}
where $\Xi^{lm} \in \R^{(L+M) \times (L+M)}$ is the covariance matrix between the averages in state $l$ and those in state $m$.
One could use, for example, the results in Chapter~17 of Ref.\@~\onlinecite{meyn_markov_2009} to verify our CLT assumption~(\ref{eq:sample_average_clt}). See Ref.\@~\onlinecite{lelievre2010free} for a more detailed discussion of the CLT in the context of molecular dynamics.

The structure of $\Xi(y)$ depends on precisely how the states are sampled. 
We are interested primarily in two particular cases: (1) The $X^j_t$ are independent Markov chains and the sample fractions $\kappa_j$ may differ but do not vary with $N$. (2) The sample fractions $\kappa_j = 1/L$ are equal and $(X^1_t, \dots, X^L_t)$ is a Markov process. The first case covers umbrella sampling or alchemical calculations performed without replica exchange. The second case covers parallel tempering and replica exchange umbrella sampling.

In the first case, since the processes sampling the different states are independent, all off-diagonal blocks of $\Xi(y)$ are zero. The diagonal blocks can be expressed as 
\begin{align}
    \Xi^{ll}_{ij}(y)  &= 
  \frac{1}{\kappa_l} \bigg( \cov \left\{ \xi_i(X_t^l,y), \xi_j(X_t^l,y) \right\}\nonumber \\
  &\quad +
  2\sum_{k=1}^\infty \cov\left\{\xi_i(X_t^l,y),\xi_j(X_{t+k}^l,y)\right\} \bigg ),
\end{align}
where here we assume that the process $X^l_t$ is in stationarity as in~(\ref{eq: first integrated autocovariance formula}). The factor of $1/\kappa_l$ arises since in~(\ref{eq:sample_average_clt}) we scale by $\sqrt{N}= \sqrt{N_l/\kappa_l}$ instead of $\sqrt{N_l}$.

In the second case, the processes sampling the states are correlated, so off-diagonal blocks may be nonzero. In this case, we have 
\begin{align}
    \Xi^{lm}_{ij}(y)  &= 
  L \bigg (\cov \left\{ \xi_i(X_t^l,y), \xi_j(X_t^m,y) \right\}\nonumber \\
  &\quad +
  2\sum_{k=1}^\infty \cov\left\{\xi_i(X_t^l,y),\xi_j(X_{t+k}^m,y)\right\}\bigg ),
\end{align}
where here we assume that the joint process $(X^1_t, \dots, X^L_t)$ is in stationarity. The factor of $L$ arises since in~(\ref{eq:sample_average_clt}) we scale by $\sqrt{N}= \sqrt{L N_l}$ instead of $\sqrt{N_l}$.

Under the assumptions discussed in Section~\ref{sec:background}, the root of $\RFN_i$ will converge to 
\begin{equation}
    v = \left(f_1 , \ldots, f_L, \omega_1, \ldots, \omega_M \right),
    \label{eq:concatenated_avgs}
\end{equation}
where we have defined
\begin{equation}
    \omega_i
    =\sum_{j=1}^L \kappa_j \< \frac{w_i(X_t^j)}{\sum_{k=1}^L \kappa_k q_k (X_t^j) e^{f_k} }\>_j.
\end{equation}
The following theorem gives the rate of this convergence.
\begin{theorem}
    Assume that when $y = v$, the central limit theorem in~\eqref{eq:sample_average_clt} holds.
    Let $A \in \R^{(L+M) \times (L+M)}$ be the matrix with entries
    \begin{equation}
        A_{jl} =  \sum_{m=1}^L \sum_{n=1}^L \kappa_m \kappa_n\ \Xi_{jl}^{mn}(v).
        \label{eq:defn_A}
    \end{equation}
    Under some technical assumptions (given in Section~I of the supplement),
    \begin{equation}
        \sqrt{N}(\bar{v}-v)\xrightarrow{d}N(0,{\Gamma} A \Gamma^T),
        \label{eq:base_defn}
    \end{equation}
where $\Gamma\in \R^{(L+M) \times (L+M)}$ is a matrix that can be expressed in block form as
\begin{equation}
    \Gamma = 
        \begin{bmatrix}
            H^\# & 0 \\
            \beta H^\# & I
        \end{bmatrix}
\end{equation}
where $I$ is the $L \times L$ identity matrix and the matrices $H \in \R^{M \times M}$ and  $\beta \in \R^{L \times M}$ 
are given by
\begin{align*}
    H_{ij}
    =&\kappa_i\left(\delta_{ij} - \< \frac{ \kappa_j q_j(x) e^{f_j} }{\sum_k \kappa_k q_k(x) e^{f_k} } \>_i \right) \\
    \beta_{ij} =&  \kappa_j \< \frac{w_{i}}{\sum_k \kappa_k  q_k(x) / z_k} \>_j 
\end{align*}
and $H^\#$ is the group inverse of $H$.
\label{lem:CLT_for_v}
\end{theorem}
A proof is given in the supplement.
In~\eqref{eq:base_defn} we have used the group inverse, a type of matrix pseudoinverse.
A numerical recipe for estimating $H^\#$ can be found in~\onlinecite{golub1986using}.

\subsection{CLTs and the Delta Method}\label{ssec:delta_method}

For most applications, practitioners are not interested in the values of $\bar v$ directly, but instead wish to evaluate nonlinear combinations of these terms.
To construct a CLT for these combinations, one can employ the Delta method.
\begin{lemma}[The Delta method; Proposition 6.2 in Bilodeau and Brenner\cite{bilodeau2008theory}]\label{lem:delta_method}
    Let $\theta_N$ be a sequence of random variables taking values in $\R^d$.
    Assume that a central limit theorem holds for $\theta_N$ with mean $\mu\in\R^d$ and an asymptotic covariance matrix $\Sigma \in R^{d \times d}$, i.e.
    \begin{equation}
        \sqrt{N}\left(\theta_N -\mu \right) \xrightarrow{D} \mathbf{N} \left(0, \Sigma \right).
    \end{equation}
    Let $\Phi:\R^d \to \R$ be a function that is differentiable at $\mu$. We then have the central limit theorem
    \begin{equation}
        \sqrt{N}\left(\Phi(\theta_N) - \Phi(\mu) \right) \xrightarrow{D} \mathbf{N}
            \left(0, \grad \Phi(\mu)^T \Sigma \grad \Phi(\mu) \right)
            \label{eq:delta_method_var}
    \end{equation}
    for the sequence of random variables $\Phi(\theta_N)$. 
\end{lemma}
To apply Lemma~\ref{lem:delta_method}, we set $\mu$ and $\theta_N$ to be $v$ and $\bar{v}$ respectively and set $\Phi$ to be the function taking $v$ to our quantity of interest.
For instance, to construct a CLT for an estimate of $\Delta G_{\alpha\to \beta}$ constructed using~\eqref{eq:mbar_state_fe_diff},
we set 
$\Phi(y)=y_L - y_1$,
and $\grad \Phi (\mu)$ is given by
\begin{equation}
    \grad \Phi (\mu) = \left(-1,0, \ldots, 0, 1, 0, \dots\right)^T.
    \label{eq:alchem_fe_gradient}
\end{equation}
We then substitute into~\eqref{eq:delta_method_var} to get the asymptotic variance of our estimate of $\Delta G_{\alpha\to\beta}$.
Similarly, for~\eqref{eq:us_average}, we set
$w_{L+1} = g q $ and $w_{L+2} = q $,
and~\eqref{eq:us_average} is recovered by setting $\Phi(y)=  \frac{y_{L+1}}{y_{L+2}}$.
Then $\grad \Phi(\mu)$ is zero apart from the last two entries, which are given by
\begin{align}
    \grad \Phi (\mu)_{L+1} =&
        \frac{1}{\omega_{L+2} } \nonumber \\
    \grad \Phi (\mu)_{L+2} =&
        -\frac{\omega_{L+1}}{\omega_{L+2}^2}
\end{align}
respectively.
As a final example, we consider the construction of error estimates of free energy differences estimated using umbrella sampling.
We set $w_{L+1}=\1_A q$ and $w_{L+2} = \1_B q $, and $\Phi(\mu) = -\log (w_{L+2}) + \log (w_{L+1})$.
Then $\grad \Phi(\mu)$ is again zero apart from the last two entries, which are
\begin{align}
    \grad \Phi (\mu)_{L+1} =&
        -\frac{1}{\omega_{L+1} } \nonumber \\
    \grad \Phi (\mu)_{L+2} =&
        \frac{1}{\omega_{L+2}}
        \label{eq:us_fe_gradient}
\end{align}

Consequently, we can combine lemma~\ref{lem:delta_method} with theorem~\ref{lem:CLT_for_v} to have a CLT for MBAR estimates.
\begin{theorem}
    Let $\Obs$ be an observable whose MBAR estimate $\bar{\Obs}$ is constructed by applying a function $\Phi:\R^{L+M} \to \R$ to the vector $\bar{v}$, and assume that $\Phi$ is differentiable at $v$.
    The estimate $\bar{\Obs}$ then obeys
    \begin{equation}
        \sqrt{N}\left(\bar{\Obs} - \Obs \right) \xrightarrow{D} N\left(0, \mathcal{A} \right).
        \label{eq:main_clt}
    \end{equation}
    where the asymptotic covariance matrix $\mathcal{A}$ is given by
    \begin{equation}
        \MBARCov = \grad \Phi^T(v)  \Gamma A \Gamma^T \grad \Phi(v) 
        \label{eq:mbar_cov_matrix}
    \end{equation}
\end{theorem}
\begin{proof}
    The proof follows immediately by applying \ref{lem:delta_method} to Theorem~\ref{lem:CLT_for_v}.
\end{proof}

\subsection{Computationally Estimating the Asymptotic Variance}\label{ssec:computational_procedure}
In practice, one could directly estimate asymptotic variances for observables by individually estimating each of the matrices and vectors in~\eqref{eq:main_clt}.
However, directly evaluating $A$ would require first populating the $\Xi$ matrix, which would in turn require evaluating as many as $L^2(L+M)^2$ correlation functions.
Consequently, we provide simplified formulas for evaluating the asymptotic variance of observables
in the specific case where sampling is performed independently in every state.
In Section~II of the supplement, we give analogous formulas for schemes such as Parallel Tempering and Replica Exchange Umbrella Sampling 
in which all states are sampled jointly using a single Markov chain.

If each state is sampled independently, then $\Xi^{lm}$ is zero for $l\neq m$, eliminating one of the sums in~\eqref{eq:defn_A}.
In Section~II of the supplement, we show that by moving the remaining sum to the outside
and bringing the remaining terms inside the expectation
we can rewrite the integrated covariance in~\eqref{eq:main_clt} as
\begin{align}
    \MBARCov= &\sum_{k=1}^L \cov \left\{ \chi_i (X_t^k), \chi_j (X_t^k) \right\}   \nonumber \\
             &+\sum_{k=1}^L 2\sum_{\tau=1}^\infty \cov\left\{\chi_i(X_t^k),\chi_j(X_{t+\tau}^k)\right\} 
            \label{eq:practical_acovar}
\end{align}
where 
\begin{equation}
    \chi_i(x) = 
    \sum_{i=1}^{L+M} \sqrt{\kappa_i} \xi_i(x, v) \left({\Gamma}^T  \grad \Phi\right)_{ij}(v)
    \label{eq:err_trajectory}
\end{equation}
To construct an estimate of the asymptotic variance, we first replace $\nu$ in~\eqref{eq:err_trajectory} with the MBAR  estimate $\bar \nu$ from sampled data and then estimate the integrated autocovariance of the resulting trajectory.
This integrated autocovariance can be estimated using standard methods. In this work we employ the ACOR algorithm\cite{acor2014}.
Moreover, since each summand in~\eqref{eq:practical_acovar} depends only on the sampling in state $i$,
we can interpret the integrated autocovariance of $\chi_i$ as accounting for how much state $i$ contributes to the total error.
A Python code implementing this algorithm for estimating asymptotic error can be found in the EMUS repository\cite{emus}.
\section{Applications}
We demonstrate our error estimator on two test cases: an alchemical free energy calculation and an umbrella sampling calculation. 
\subsection{Alchemical calculation of the free energy of solvating methane in water}
The solvation free energy of methane can be determined via an alchemical simulation process in which the interaction between methane and a bath of water molecules is introduced gradually.
We interpolate between the two states using \eqref{eq:alchemical_interpolation}, setting $H_\alpha$ to the Hamiltonian where the methane molecule and the water do not interact, and $H_\beta$ to the Hamiltonian where they interact fully.
We then estimate the free energy difference between the two states using~\eqref{eq:mbar_state_fe_diff} and estimate the asymptotic variance as described in Subsection~\ref{ssec:computational_procedure}.

We performed 20 independent alchemical simulations at 298 K using GROMACS version 2019.4\cite{lindahl_2019_3460414}, the OPLS-AA force field\cite{robertson2015improved}, and the TIP3P water model\cite{jorgensen1983comparison}. A total of 21 equidistant $\lambda$ values from 0 to 1 (endpoints included) were chosen. Each state was equilibrated at constant volume and then at constant pressure of 1 bar for 100 ps using the Parinello-Rahman barostat with a time constant of 1 ps.  The state was then further sampled at constant pressure for 1 ns to generate 1000 data points.  The P-LINCS algorithm was used to constrain bonds to hydrogen atoms\cite{hess1997lincs,hess2008plincs}. In all simulations a stochastic Langevin dynamics integrator with a time step of 2 fs and time constant of 1 ps was used to maintain a constant temperature of 300 K.
In Figure~\ref{fig:alchem_fe}, we plot the cumulative free energy change between states as well as the free energy difference between successive states.
\begin{figure}
  \includegraphics[width=1\linewidth]{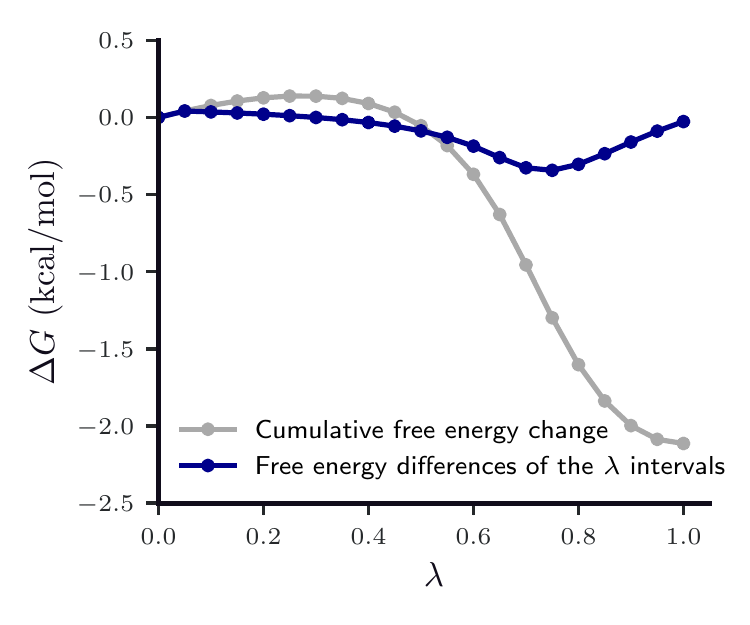}
  \caption{The free energy of solvating methane in water computed from alchemical simulations. The blue line indicates free energy differences between neighboring states, and the gray line is the cumulative free energy changes. The total free energy of solvation, i.e., the cumulative $\Delta G$ at $\lambda=1$, is estimated to be $2.13$~{kcal/mol}.}
  \label{fig:alchem_fe}
\end{figure}

\begin{figure*}
    \includegraphics[scale=1.0]{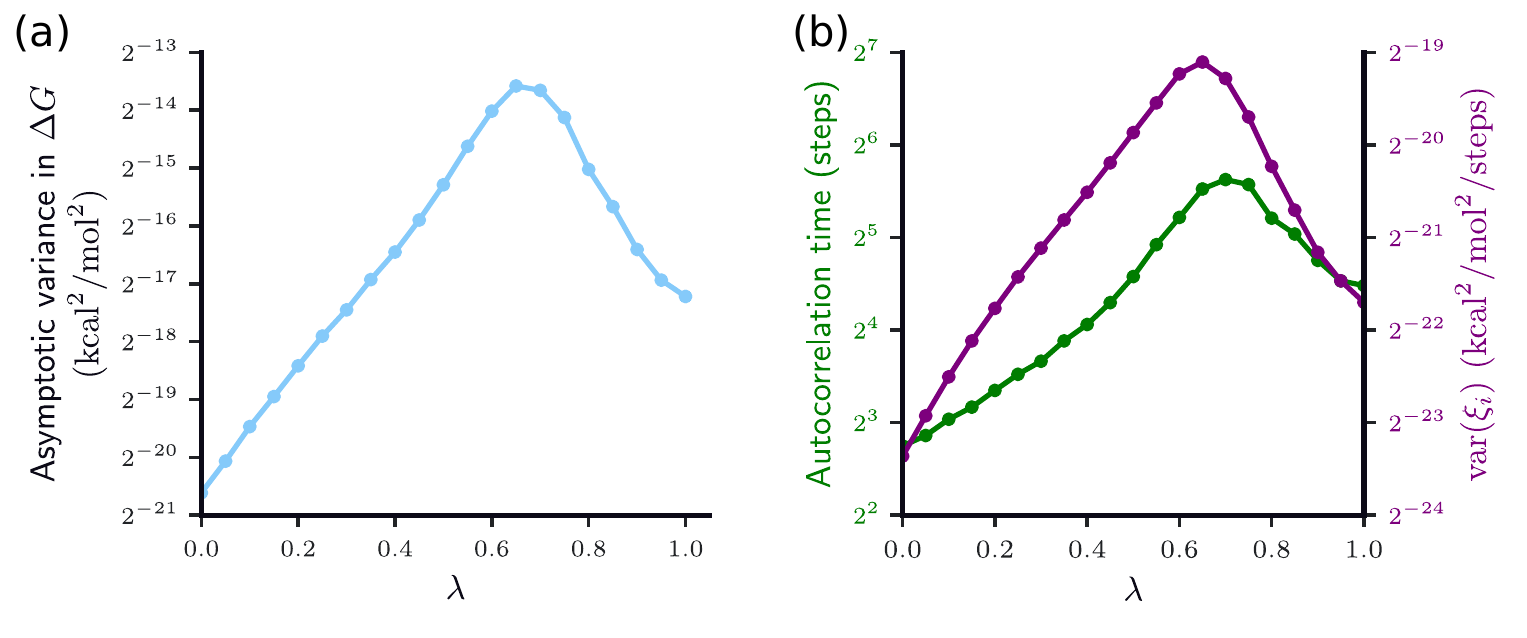}
    \caption{Analysis of the estimated error in the free energy of solvation for methane.
        (a) State contributions to the asymptotic variance of the free energy difference between the initial ($\lambda=0$) and final states ($\lambda=1$).
        (b) Breakdown of the error contributions into the integrated autocorrelation times and the variance of $\xi_i$.
        For ease of comparison, quantities are plotted on a logarithmic scale.
        }
    \label{fig:alchem_err}
\end{figure*}
The total asymptotic standard deviation in the solvation free energy is estimated to be $0.0216\pm0.0005$ kcal/mol using~\eqref{eq:practical_acovar} with $\grad \Phi$ given by~\eqref{eq:alchem_fe_gradient}.
This is close to the standard deviation over all trials which is equal to $0.0225$ kcal/mol. 
The error contributions from all states are shown in Figure \ref{fig:alchem_err}a.
As the error contributions for different states can vary by more than two orders of magnitude, 
we have chosen to depict them on a logarithmic scale.
Moreover, comparing with Figure~\ref{fig:alchem_fe} we see that the error contributions correlate with the magnitudes of the free energy differences between neighboring states.
The fact that different states' error contributions differ by orders of magnitudes suggests that the error in alchemical free energy simulations may be dominated by a few states.
Authoritatively establishing this hypothesis would require further investigation over many alchemical simulations in a variety of settings.
However, if similar phenomena do hold generically for other alchemical simulations,
it may be possible to use error estimates to tune simulation parameters to achieve dramatic reductions in the error of MBAR estimates.
Indeed, concurrent work that attempts to allocate sampling for alchemical simulations on the fly\cite{predescu2021times} suggests that 
better allocation of computational resources can substantially reduce the error in alchemical free energy simulations.

To further examine the source of the errors in our simulation, we attempt to disentangle the effect of the dynamics used to sample the state from effects inherent to the state definition.
Recalling the definition of the integrated autocorrelation time in \eqref{eq:defn_actime} and combining it with \eqref{eq:practical_acovar}, we can further write the integrated autocovariance of each state as a product of the integrated autocorrelation time and a sampler-independent factor, namely, $\var\{\zeta_i\}$. 
In Figure~\ref{fig:alchem_err}b we plot both of the error components on a log scale: the logarithm of a state's total contribution is a sum of the two curves.  
Our results show that capturing both the sampler-independent component of the error and the integrated autocorrelation time are important for estimating the total error contribution.
Indeed, previous work has typically focused on optimizing the state parameters using solely thermodynamic 
properties; for instance, \onlinecite{shenfeld2009minimizing,pham2011identifying} for instance,  used information-geometric distances between states.
However, our results suggest that to fully capture all sources of error, such approaches must also take into account kinetic
effects from the specific choice of sampler used.
This corrobates previous work\cite{pham2012optimal} which has observed that the thermodynamical optimal choice of alchemical states may not be optimal in practice due to the resulting states having exceedingly large correlationtimes.

\subsection{Umbrella Sampling Simulation of the Alanine Dipeptide}
\begin{figure}[t]
\centering
  \includegraphics[width=1\linewidth]{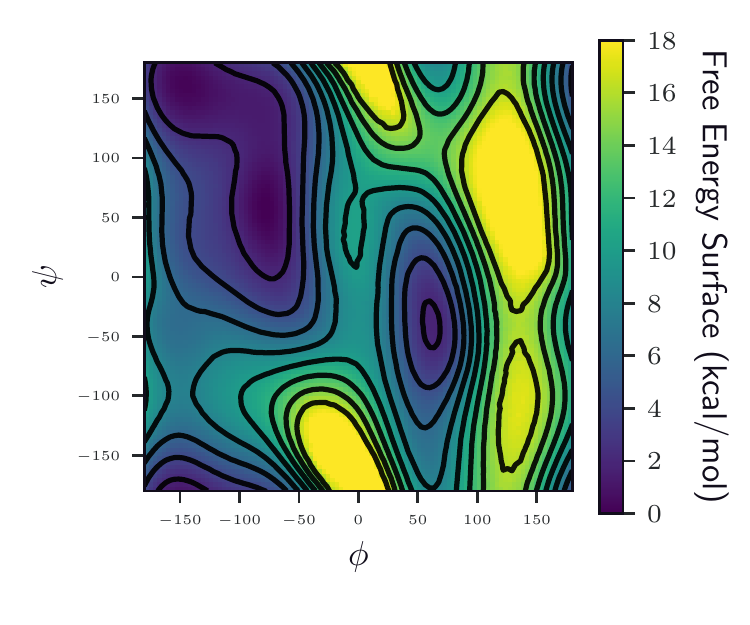}
  \caption{Free energy as a function of the $\phi$ and $\psi$ dihedral angles (measured in degrees) of the alanine dipeptide. The scale bar indicates free energy values in kcal/mol, and the contour spacing is $2$ kcal/mol. }
  \label{PMF}
\end{figure}
We also applied the error estimator to a two-dimensional umbrella sampling simulation of the alanine dipeptide ($N$-acetyl-alanyl-$N'$-methylamide) in vacuum.
We performed 15 independent umbrella sampling calculations for the free energy as a function of the $\phi$ and $\psi$ dihedral angles. 
Simulations were run at 300 K using GROMACS version 2019.4\cite{lindahl_2019_3460414} with harmonic restraints applied to  $\phi$ and $\psi$ using the PLUMED 2.5.3 software package\cite{tribello2014plumed}. The molecule was represented by the AMBER force field with bonds to hydrogen
atoms constrained by the LINCS algorithm\cite{hess1997lincs}. The force constant for the harmonic restraints was $0.0018$ $\text{kcal}$ $\text{mol}^{-1}$ $\text{degree}^{-2}$, which corresponds to a Gaussian bias function with a standard deviation of $\ang{18}$ in the absence of the molecular potential. 
We partitioned each dihedral angle into 30 intervals and placed the centers of the harmonic restraints at the centers of the cells of the resulting $30\times 30$ grid; the resulting grid ranged from $(-\ang{171}, -\ang{171})$ to $(\ang{171}, \ang{171})$.
Each state was sampled independently using the velocity Langevin dynamics integrator in GROMACS with a time step of 2 fs and a time constant of 0.1 ps. Each state was equilibrated for 40 ps and then sampled for 10 ms, with $\phi$ and $\psi$ values output every 0.4 ps. 

\begin{figure*}
    \includegraphics[scale=1.0]{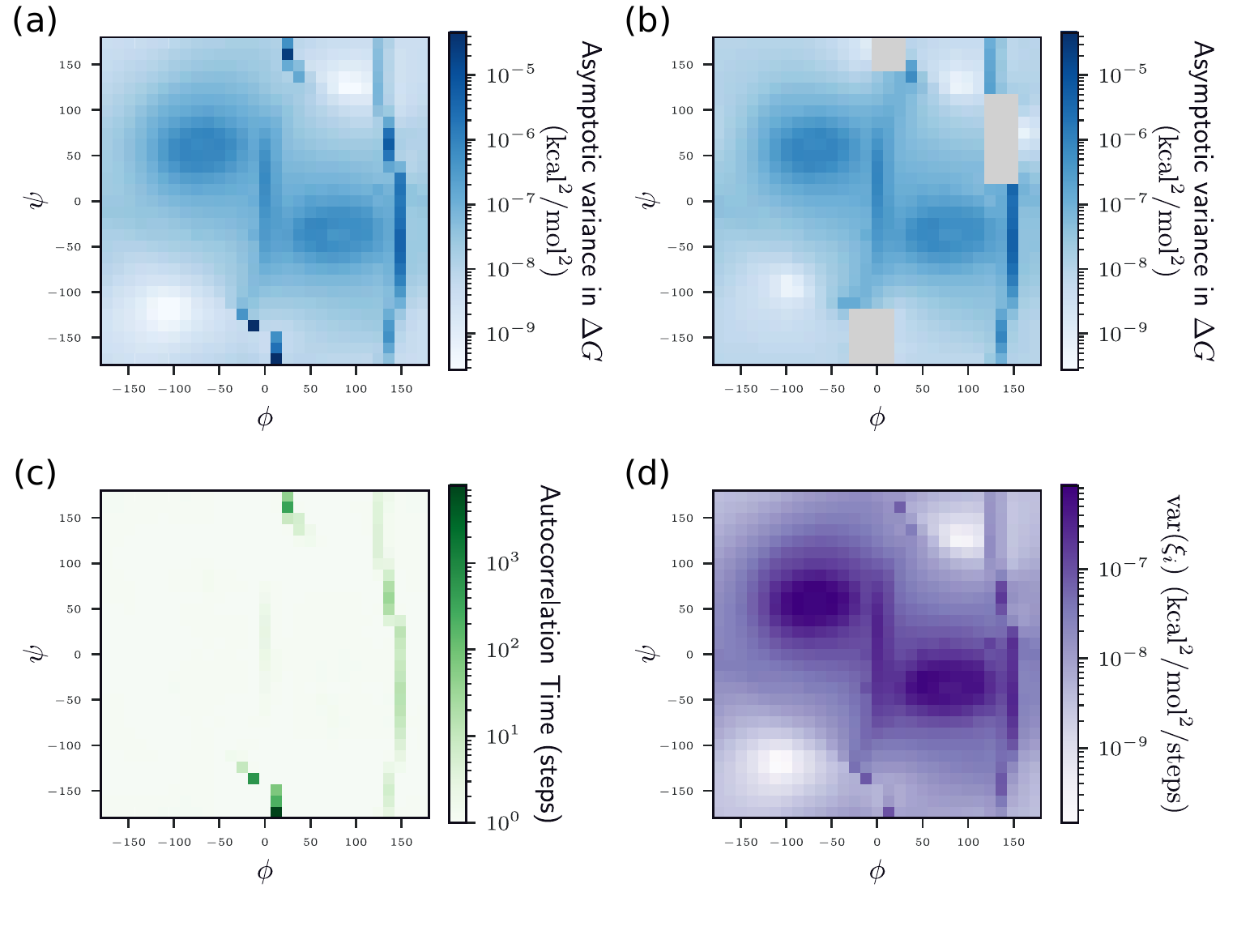}
    \caption{Analysis of the estimated error in the free energy difference between the C$_{\rm 7ax}$ and C$_{\rm 7eq}$ basins of the alanine dipeptide, shown on a logarithmic scale.
        (a) State contributions to the asymptotic variance of the free energy calculated using all states.
        (b) State contributions to the asymptotic variance of the free energy calculated using a set of states curated by analyzing the estimator for the asymptotic variance. States removed from the dataset are indicated in gray.
        (c) Integrated autocorrelation times of $\xi_i$ for the free energy estimate using every state.
        (d) Variance of $\xi_i$ for the free energy estimate using every state.
        }
        \label{fig:us_errs}
\end{figure*}

We used the MBAR equations to calculate the free energy as a function of $\phi$ and $\psi$ and, for each trial, the free energy difference between the C$_{\rm 7ax}$ and C$_{\rm 7eq}$ basins. 
We obtain the latter from the logarithm of the ratio of averages of two indicator functions. 
The C$_{\rm 7ax}$ basin is defined as the region in the $\phi\psi$-space enclosed by a circle of radius $\ang{10}$ centered at  $(\ang{65}, \ang{-40})$. The C$_{\rm 7eq}$ basin is similarly defined by the circle of radius $\ang{10}$ centered at $(\ang{-75}, \ang{50})$.

In Table~\ref{tab:asy_err} we give our estimate of the error in the free energy difference evaluated as the asymptotic standard deviation (square root of the asymptotic error), as evaluated by~\eqref{eq:practical_acovar}, with $\grad Phi$ given by~\eqref{eq:us_fe_gradient}.  For comparison, we also give an estimate of the standard deviation calculated over 15 identical replicates.
Our error estimate underestimates the true error by a factor of three, which we attribute to the difficulty in estimating the autocorrelation time accurately in the states.
Comparing the plots of the integrated autocorrelation times in Figure~\ref{fig:us_errs}c and the error contributions in Figure~\ref{fig:us_errs}d 
to the free energy surface in Figure~\ref{PMF}, we see that the error estimate is dominated by states in high-free energy regions that have long autocorrelation times.
These states are located in regions where the alanine dipeptide is highly contorted,
causing the free energy to increase sharply and complicating the task of sampling.
For instance, states near free energy maxima can be bifurcated by a steep peak in the free energy.
Sampling the state requires rare barrier-crossing events resulting in long autocorrelation times.

The fact that these states contribute considerably to our error estimate confirms the intuition that sampling unphysical regions can pollute free energy estimates.
Indeed, this intuition has been the basis of a previous adaptive umbrella sampling algorithm \cite{wojtas2013self}. 
To validate this idea, we removed several high-free energy states and recalculated both the estimate of the asymptotic error in the free energy difference as well as the standard deviation over the 15 replicates.
We give the results in Table~\ref{tab:asy_err} and plot the resulting error contributions in Figure~\ref{fig:us_errs}b with the removed states marked in gray.
We see that the free energy estimate with high-free energy states filtered out has a lower error by a factor of three, despite the decrease in the amount of data.
While we leave a systematic procedure for removing states for future work, 
our preliminary results, along with our earlier work on the Eigenvector Method for Umbrella Sampling and its application \cite{thiede2016eigenvector,dinner2020stratification,antoszewski2020insulin}, demonstrate that our error estimates have the potential to
improve the error of umbrella sampling and other multistate methods.

\begin{table*}[t]
\caption{Asymptotic standard deviation of the free energy between the C$_{\rm 7ax}$ and C$_{\rm 7eq}$ basins of the alanine dipeptide compared with the standard deviation over 15 independent simulations. The top row uses all states in the data set, while the bottom row uses insight from the error estimator to remove problematic states.}
\begin{tabularx}{0.8\textwidth}{YYY}
\toprule
States Used & Estimated Asymptotic SD  &  SD over Replicates\\
 & (kcal/mol)  &   (kcal/mol) \\

\midrule
All & $0.0308\pm 0.0063$     &   $0.1125$              \\
Curated  & $0.0100\pm 0.0001$               & $0.0377$  \\\hline
\end{tabularx}
\label{tab:asy_err}
\end{table*}

\section{Conclusions}
We derive a central limit theorem for estimates of both the normalization constants and function averages of the MBAR estimator. 
The central limit theorem allows us to devise a computational procedure for estimating the asymptotic error for arbitrary observables calculated through MBAR.
In particular, it allows us to estimate the asymptotic error in free energy calculations. 
Notably, if states are sampled independently, the analytical expression of the total asymptotic error takes the form of a sum of contributions from all states.  
This enables us to trace how error in individual Markov chains contribute to the total error.
Moreover, unlike previous error estimators, our approach explicitly accounts for the effects of correlation within each sampled trajectory.

We demonstrate the error estimator for an alchemical calculation of the solvation free energy of methane and a two-dimensional umbrella sampling calculation of the free energy of rearrangement of the alanine dipeptide. 
In both cases, the asymptotic error estimates agree reasonably well with the true error over all replicates. 
Moreover, for both our alchemical calculation and the umbrella sampling calculation we observe that the differences in the error contribution between states correlates strongly with the autocorrelation time of the associated states.
For instance, for our umbrella sampling calculation, we observe that the error contributions are dominated by states in high-free energy regions in which dynamical variables decorrelate slowly.
Upon the removal of those states, we see a significant error reduction.
These results highlight the importance of error analysis that accounts for correlation in a sampled trajectory.
We hope to investigate adaptive sampling strategies based on our error estimates in future work.
\section{Acknowledgements}
We wish to thank Jonathan Weare for advice and for pointing out a simplification in our error analyses, and Adam Antoszewski for helpful discussions.  This work was supported by National Institutes of Health award R35 GM136381 and National Science Foundation awards DMS-2054306 and DMS-2012207.  The Flatiron Institute is a  division of the Simons Foundation.


\bibliography{aipsamp}

\pagebreak
\widetext
\begin{center}
    \textbf{\large Supplement for Understanding the Sources of Error in MBAR through Asymptotic Analysis}
\end{center}
\setcounter{equation}{0}
\setcounter{figure}{0}
\setcounter{theorem}{0}
\setcounter{table}{0}
\setcounter{section}{0}
\setcounter{page}{1}
\makeatletter
\renewcommand{\theequation}{S\arabic{equation}}
\renewcommand{\thesection}{S\arabic{section}}
\renewcommand{\thefigure}{S\arabic{figure}}

\section{Derivation of the Asymptotic Variance of MBAR}

We prove that the MBAR estimator is asymptotically normal, and we give a convenient formula for the asymptotic variance.
Our argument is based on the work of~\onlinecite{geyer1994estimating}.
Our contribution is to fill in details missing from Geyer's manuscript and to correct errors.
Most importantly, the formula for the asymptotic variance of observable averages $\< g \>$ is not correct in the Geyer's manuscript.

\subsection{The Maximum (Pseudo)likelihood Perspective}

Following~\onlinecite{geyer1994estimating}, we first introduce the MBAR estimate as the maximizer of a pseudolikelihood.
    We introduce the (absolute) free energy,
\begin{equation*}
  \eta_j = - \log(c_j)
\end{equation*}
Define for any $\VFS \in \R^L$,
\begin{align}
  \qmix (x,\VFS) &= 
              \sum_{k=1}^L \kappa_k q_k(x) \exp(\VFS_k), \text{ and } 
              \label{eq:defn_mixture_dist}
              \\
  p_j(x,\VFS) &=\frac{ \kappa_j q_j(x) \exp(\VFS_j)}{\qmix(x,\VFS)}.
\end{align}
We can then define the pseudolikelihood
\begin{equation*}
  l_N(\VFS) = \frac{1}{N} \sum_{j=1}^L \sum_{t=1}^{N_j} \log p_j(X^j_t,\VFS)
\end{equation*}
based on a formal similarity with a logistic regression model.
The MBAR estimating equation for $\bar \eta$ is equivalent with the first order optimality condition
\begin{equation}\label{eq: optimality}
  \nabla l_N (\bar \eta) = 0
\end{equation}
for this pseudolikelihood. 
We refer to~\cite{gill_large_1988,vardi1985empirical,geyer1994estimating,kong2003theory} for more explanation of this perspective and for detailed derivations.
Here, we merely observe that if samples are collected IID in each state then $l_N$ is the likelihood of the data given a particular value $\VFS$ of the free energy. We note that this is not true for other sampling strategies where the samples are not IID, but one can still use~\eqref{eq: optimality} to estimate $\eta$ even in that case.
When written in maximum likelihood form and in terms of the variable $\eta$, the estimating equations have a very convenient structure. In particular, $l_N$ is concave and the Fisher information $\nabla^2 l_N$ is a generator matrix. We will exploit aspects of this structure below.

Observe that $l_N$ is unchanged if one adds a constant to $\VFS$, so~\eqref{eq: optimality} determines $\bar \eta$ only up to an additive constant. 
Indeed, we chose to work with estimates of the relative free energies, which obey the constraint
\begin{equation}
  \label{eqn: constraint}
  \sum_{k=1}^L \bar f_k=0.
\end{equation}
With this particular constraint, $\bar f$ is not exactly an estimator of $\eta$, since we may have $\sum_{j=1}^L \eta_j \neq 0$. 
Instead, $\bar f$ converges to
\begin{equation*}
  f = \eta - \sum_{k=1}^L \eta_k.
\end{equation*}
We note that this arbitrary normalization of the free energies does not cause difficulties, since the MBAR formulas for observable averages are unchanged if one adds a constant to $\bar f$ and only free energy differences are important for most applications.

When the constraint~\eqref{eqn: constraint} is imposed, one can show that if the matrix
\begin{equation*}
\bar M_{ij} = \frac{1}{N_j} \sum_{t=1}^{N_j} q_i(X^j_t)
\end{equation*}
is irreducible, then~\eqref{eq: optimality} determines a unique $\bar f$, cf.\@ Theorem 1 in Ref.\@~\onlinecite{geyer1994estimating} and Theorem~1.1 in Ref.\@~\onlinecite{gill_large_1988}. Moreover, if $M_{ij} = \langle q_i \rangle_j$ is irreducible, then (a) with probability one $\bar M$ is irreducible for sufficiently large $N$ (cf.\@ Theorem~1 in Ref.\@~\onlinecite{geyer1994estimating} or Proposition~1.1 in Ref.\@~\onlinecite{gill_large_1988}) and (b) $\bar f \xrightarrow{as} f$ (cf.\@ Theorem~1 in Ref.\@~\onlinecite{geyer1994estimating}).  
When  $M_{ij} = \langle q_i \rangle_j$ is irreducible, we say that the densities $q_1, \dots, q_L$ are inseparable.
We will assume inseparability throughout the remainder of this work:
\begin{assumption}\label{asm: inseparable}
  We assume that the densities $q_1, \dots, q_L$ are inseparable, i.e.\@ the matrix $M_{ij} = \langle q_i \rangle_j$ is irreducible.
\end{assumption}

Once estimates of the free energies $\bar f$ have been constructed, averages of observables are estimated in terms of quantities of the form 
\begin{equation}\label{eq: definition of bar omega}
    \bar{\omega}_i (\bar f)
    = \sum_{k=1}^L \kappa_k \frac{1}{N_k} \sum_{t=1}^{N_k} \frac{w_i(X^k_t)}{\qmix(X^k_t,\bar f)}
\end{equation}
for some functions $w_1, \dots, w_M$, as explained in the text. 

\subsection{Outline of Proof of Asymptotic Normality}

We will prove asymptotic normality of $(\bar f, \bar \omega(\bar f))$.
To do so, it is expedient to define the function $\bar \RFN: \R^L \times \R^M \rightarrow \R^L \times \R^M$ by 
\begin{equation*}
  \bar \RFN(\VFS,\VAS) = (\nabla l_N(\VFS),\VAS - \bar \omega(\VFS)), 
\end{equation*}
where $\bar \omega(\VFS)$ is defined by~(\ref{eq: definition of bar omega}) but with $\VFS$ in place of $\bar f$.
Under Assumption~\ref{asm: inseparable}, $(\bar f,\bar{\omega}(\bar f))$ is the unique root of $\bar \RFN$ for $N$ sufficiently large. We will proceed as follows: First, we show that for any fixed $(\VFS,\VAS) \in \R^L \times \R^M$, a LLN and CLT hold for $\bar \RFN(\VFS,\VAS)$, so
\begin{align*}
  \bar \RFN(\VFS,\VAS) &\xrightarrow{as} \RFN(\VFS,\VAS)
\end{align*}
and
\begin{align*}
\sqrt{N} (\bar \RFN(\VFS,\VAS) - \RFN(\VFS,\VAS) ) &\xrightarrow{d} N (0, A(\VFS,\VAS))
\end{align*}
for some limiting value $\RFN(\VFS,\VAS)$ and covariance matrix $A(\VFS,\VAS)$. Second, we compute a Taylor expansion of roughly the form
\begin{equation*}
    \bar \RFN(\bar f, \bar{\omega}(\bar f)) - \bar \RFN(f,\omega) = - \bar \RFN(f,\omega) = \bar \RFN'(f,\omega) ((\bar f, \bar{\omega}(\bar f)) - (f,\omega)) + \text{ higher order terms},
  \end{equation*}
where $\omega$ is the limiting value of $\bar \omega(f)$ defined in~\eqref{eqn: limits of bar w}.
Third, we derive a matrix $\bar \Gamma$ that acts as a generalized inverse for $\bar \RFN'$ in the Taylor expansion, which gives
\begin{equation*}
  - \bar \Gamma \bar \RFN(f, \omega) = (\bar f, \bar{\omega}(\bar f)) - (f, \omega) + \text{ higher order terms}.
\end{equation*}
We show that $\bar \Gamma$ converges almost surely to a limit $\Gamma$, and we give an explicit formula for $\Gamma$.
Finally, we conclude by Slutsky's theorem and the CLT for $\bar \RFN(\VFS,\VAS)$ that 
\begin{equation*}
  \sqrt{N} ((\bar f, \bar{\omega}(\bar f)) - (f, \omega)) = -\bar \Gamma \sqrt{N} (\bar \RFN(f, \omega) - \RFN(f, \omega)) + \text{higher order terms} \xrightarrow{d} N(0, \Gamma A(f,\omega)\Gamma^t).
\end{equation*}

\subsection{Derivatives of $l_N$ and $\bar \omega$}
\label{sec: taylor}

The proof outlined above requires the Taylor expansion of $\bar \RFN$ at $(f, \omega)$. We begin by computing the derivatives of $l_N$ and also some upper bounds on derivatives that we will use to estimate higher order terms in our expansion.
We observe that
\begin{align*}
  \frac{\partial p_j(x,\VFS)}{\partial y_l} &= \frac{\partial}{\partial y_l} \frac{\kappa_j q_j(x) \exp(\VFS_j)}{\sum_{k=1}^L \kappa_k q_k(x) \exp(\VFS_k)} \\
  &=  p_j(x,\VFS) (\delta_{lj} - p_l(x,\VFS)).  
\end{align*}
Therefore,
\begin{align*}
  \frac{\partial l_N}{\partial y_l}(\VFS) &= \kappa_l - \frac1N \sum_{k=1}^L \sum_{t=1}^{N_k}  p_l(X^k_t,\VFS), \\
  \frac{\partial l_N}{\partial y_l \partial y_m}(\VFS) &=
  \begin{cases}
    - \frac1N \sum_{k=1}^L \sum_{t=1}^{N_k}  p_l(X^k_t,\VFS)(1- p_l(X^k_t,\VFS)) &\text{ if } l = m \\
    \frac1N \sum_{k=1}^L \sum_{t=1}^{N_k}  p_l(X^k_t,\VFS)p_m(X^k_t,\VFS) &\text{ if } l \neq m.
  \end{cases}                                                                     
\end{align*}
Moreover, one sees immediately from the analogous formula for third derivatives (and using $0 \leq p_l(x, \VFS) \leq 1$) that
\begin{equation}
  \label{eqn: bound on third derivatives}
\left \lvert \frac{\partial l_N}{\partial y_l \partial y_m \partial y_p}(\VFS) \right \rvert \leq 2.
\end{equation}
This will allow us to control the error terms of the Taylor expansion.

We also require derivatives of the $\bar{\omega}_i$'s. Define for any $\VFS \in \R^L$,
\begin{equation*}
  \bar{\omega}_i(\VFS) = \frac1N \sum_{k=1}^L \sum_{t=1}^{N_K} \frac{w_i(X^k_t)}{\qmix(X^k_t,\VFS)}.
\end{equation*}
We have
\begin{align*}
  \frac{\partial \bar{\omega}_i}{\partial y_l}(\VFS) &= -\frac{1}{N} \sum_{k=1}^L \sum_{t=1}^{N_K} \frac{w_i(X^k_t)}{\qmix(X^k_t,\VFS)} p_l(X^k_t,\VFS), \\
  \frac{\partial \bar{\omega}_i}{\partial y_l y_m}(\VFS)
                                                 &= \frac{1}{N} \sum_{k=1}^L \sum_{t=1}^{N_K} \frac{w_i(X^k_t)}{\qmix(X^k_t,\VFS)} ( p_l(X^k_t,\VFS) p_m(X^k_t,\VFS) - p_l(X^k_t,\VFS) (\delta_{lm} - p_m(X^k_t,\VFS) )\\
  &= \frac{1}{N} \sum_{k=1}^L \sum_{t=1}^{N_K} \frac{w_i(X^k_t)}{\qmix(X^k_t,\VFS)}  2p_l(X^k_t,\VFS)( p_m(X^k_t,\VFS) - \delta_{lm}).
\end{align*}
Note that the second order partial derivatives are bounded independently of $N$:
\begin{equation}
  \label{eqn: bound on second derivatives of w}
   \left \lvert \frac{\partial \bar{\omega}_i}{\partial y_l \partial y_m}(\VFS) \right \rvert \leq 2 \max_x \left \lvert \frac{w_i (x)}{\qmix(x, \VFS)} \right \rvert.
\end{equation}
To control the error terms in our Taylor approximation, we will assume that the right hand side of~\eqref{eqn: bound on second derivatives of w} is bounded uniformly in $x$ for fixed $\VFS$.

\begin{assumption}\label{asm: boundedness}
  For any $\VFS \in \R^L$, there exists $C(\VFS) >0$ so that 
  \begin{equation*}
     \max_x \left \lvert \frac{w_i (x)}{\qmix(x, \VFS)} \right \rvert \leq C(\VFS)
  \end{equation*}
  for all $i = 1, \dots, M$.
\end{assumption}

We expect Assumption~\ref{asm: boundedness} to hold in practice for both umbrella sampling and alchemical calculations.
First, consider umbrella sampling, where $w_i = g_i q$ for some observables $g_i$ and $q_k = \psi_k q$ for some biasing functions $\psi_k$. Typically, the observables are bounded, so for some $D>0$,
\begin{equation}\label{eq: boundedness us case 1}
 \lvert g_i (x) \rvert \leq D 
\end{equation}
for all $x$ and $i=1, \dots, M$. In addition, we may assume that for some $E>0$,
\begin{equation}\label{eq: boundedness us case 2}
   \left \lvert \frac{1}{\sum_k \psi_k(x)} \right \rvert \leq E
 \end{equation}
 for all $x$.
We note that~\eqref{eq: boundedness us case 2} holds if the biasing functions $\psi_i$ are a partition of unity and also in the other cases considered in~\onlinecite{thiede2016eigenvector} and~\cite{dinner2020stratification}. If~\eqref{eq: boundedness us case 1} and~\eqref{eq: boundedness us case 2} hold, then Assumption~\ref{asm: boundedness} holds with
 \begin{equation*}
   C(\VFS) = \frac{D E}{\min_k \exp(\VFS_k)}.
 \end{equation*}
 In the alchemical case, if we are only interested in a free energy difference, then Assumption~\ref{asm: boundedness} is irrelevant since we need only estimate $z_1/z_L$. If we also want the averages of an observable $g$ over the densities $q_i$, then we choose $w_i=gq_i$ and Assumption~\ref{asm: boundedness} holds with
 \begin{equation*}
   C(\VFS) = \frac{\max_x \lvert g(x) \rvert}{\min_k \exp(\VFS_k)},
 \end{equation*}
 assuming that $g$ is bounded.

\subsection{Ergodicity and the Central Limit Theorem for $\bar \RFN$} 

 Observe that for any fixed $\VFS \in \R^L$, $l_N(\VFS)$, $\bar{\omega}(\VFS)$, and their derivatives are all linear combinations of ergodic averages of fixed functions over the processes $X^i_t$. We assume that these ergodic averages converge as $N \rightarrow \infty$.
 
\begin{assumption}\label{asm: ergodicity}
  For any bounded measurable $g: \R^n \rightarrow \R^m$, we have
  \begin{equation*}
    \frac{1}{N} \sum_{j=1}^L \sum_{t=1}^{N_j} g(X^j_t) \xrightarrow{as} \E[g(X)],
  \end{equation*}
  where $X$ is a random variable with density proportional to $\qmix(\cdot, f)$.
\end{assumption}

If $\kappa_i = N_i/N$ is independent of $N$,
Assumption~\ref{asm: ergodicity} holds as long as each of the chains $X^j_t$ is ergodic. (In fact, it is enough that $\limsup_{N \rightarrow \infty} \kappa_i >0$ for all $i$. We will not discuss this more general possibility, but see the assumptions made in Ref.\@~\cite{geyer1994estimating} and Ref.\@~\cite{gill_large_1988}.)

Observe that for any fixed $\VFS \in \R^L$, $p_l(x, \VFS)$ is a bounded and measurable function of $x$. 
Therefore, by Assumption~\ref{asm: ergodicity} and the formulas derived in Section~\ref{sec: taylor}, for any fixed $\VFS \in \R^L$, 
\begin{equation}
\label{eqn: limits of derivatives of likelihood} 
  \begin{split}
    \frac{\partial l_N}{\partial f_l}(\VFS) \xrightarrow{as} G_l (\VFS) &= \kappa_l - \E[p_l(X,\VFS)],\\
    \frac{\partial l_N}{\partial f_l \partial f_m}(\VFS) \xrightarrow{as} H_{lm} (\VFS) &=                                                                                  \begin{cases}                                                                          - \E [p_l(X,\VFS)(1- p_l(X,\VFS))] &\text{ if } l = m \\                              \E[ p_l(X,\VFS)p_m(X,\VFS) ] &\text{ if } l \neq m,                            \end{cases}
  \end{split}
\end{equation}
where $X$ is a random variable with density proportional to $\qmix(\cdot, f)$.
Similarly,
\begin{equation}
\label{eqn: limits of bar w}
  \begin{split}
  \bar{\omega}_i(\VFS) \xrightarrow{a.s.} \omega_i(\VFS) &= \E \left [ \frac{w_i(X)}{\qmix(X,\VFS)}\right ],\\
  \frac{\partial \bar{\omega}_i}{\partial f_l}(\VFS) \xrightarrow{a.s.} dw_{il}(\VFS) &= \E \left [ \frac{w_i(X)}{\qmix(X,\VFS)} p_l(X,\VFS)\right ],
\end{split}
\end{equation}
where again $X$ is a random variable with density proportional to $\qmix(\cdot, f)$. It follows that
\begin{equation*}
  \bar \RFN(\VFS,\VAS) \xrightarrow{as} \RFN(\VFS,\VAS) = (G(\VFS),\VAS-\omega(\VFS)). 
\end{equation*}

We also assume that a central limit theorem holds.
\begin{assumption}\label{asm: clt}
  For any fixed $(\VFS,\VAS) \in \R^L \times \R^M$,
  \begin{align*}
\sqrt{N} (\bar \RFN(\VFS,\VAS) - \RFN(\VFS,\VAS) ) &\xrightarrow{d} N (0, A(\VFS,\VAS))
  \end{align*}
  for some covariance matrix $A(\VFS,\VAS) \in \R^{(L + M) \times (L + M)}$.
\end{assumption}

Assumption~\ref{asm: clt} is a consequence of assumption~(39) in the main text.
When assumption~(39) holds, $A(f, \omega)$ is the same as the matrix $A$ defined in the main text in the statement of Theorem~IV.1.

\subsection{The Linearization of $\bar \RFN$ and its Generalized Inverse}
\label{sec: linearization}

We will transform the central limit theorem for $\bar \RFN$ to a central limit theorem for $(\bar f,\bar{\omega}(\bar f))$ by applying a generalized inverse of the linearization of $\bar \RFN$ at $(f, \omega)$ . The basic idea is very similar to the delta method. In this subsection, we derive the appropriate linearization and generalized inverse.

Define
\begin{equation*}
  \bar v = (\bar f,\bar{\omega}(\bar f)) \text{ and } v = ( f, \omega).
\end{equation*}
We have
\begin{equation*}
  \bar \RFN(\bar v) - \bar \RFN(v) = -\bar \RFN(v) =  -(\bar \RFN(v)-\RFN(v)) = \bar B (\bar v- v),
\end{equation*}
where
\begin{equation*}
  \bar B = \int_{s=0}^1 \bar \RFN'(s \bar v + (1-s) v) \, ds.
\end{equation*}
(Here, we use that $\bar \RFN(\bar v) = \RFN(v)=0$.)
We will demonstrate the existence of a generalized inverse $\bar \Gamma$ of $\bar B$ so that
\begin{equation*}
  \bar \Gamma \bar B (\bar v - v) = \bar v - v,
\end{equation*}
and we will show that $\bar \Gamma$ converges almost surely to a certain matrix $\Gamma$.

To devise the right $\bar \Gamma$, we first observe that 
\begin{equation*}
  \bar \RFN'(\VFS,\VAS) =
  \begin{pmatrix}
    \nabla^2 l_N(\VFS) & 0 \\
    -\frac{\partial \bar{\omega}}{\partial y }(\VFS,\VAS) & I
  \end{pmatrix}.
\end{equation*}
Here, $\nabla^2 l_N(\VFS) \in \R^{L \times L}$ denotes the Hessian matrix of $l_N$, $\frac{\partial \bar{\omega}}{\partial y} \in \R^{M \times L}$ with $\left (\frac{\partial \bar{\omega}}{\partial y} \right )_{ij} = \frac{\partial \bar{\omega}_{i-L}}{\partial y_j}$, and $I \in \R^{M \times M}$ is the identity matrix.

We observe that for any value of $x$,
\begin{equation*}
  -\nabla^2 l_N(\VFS) = I - \bar P(\VFS),
\end{equation*}
where $\bar P$ is the symmetric, stochastic matrix
\begin{equation*}
  \bar P_{lm}(y) =
   \begin{cases}
    1 - \frac1N \sum_{k=1}^L \sum_{t=1}^{N_k}  p_l(X^k_t,\VFS)(1- p_l(X^k_t,\VFS)) &\text{ if } l = m \\
    \frac1N \sum_{k=1}^L \sum_{t=1}^{N_k}  p_l(X^k_t,\VFS)p_m(X^k_t,\VFS) &\text{ if } l \neq m.
  \end{cases}
\end{equation*}
That is, $-\nabla^2 l_N(\VFS)$ is a generator matrix.
To see that $\bar P(\VFS)$ is indeed stochastic, first observe that
\begin{equation*}
  \bar P_{ii}(\VFS) \geq \frac34,
\end{equation*}
since we have $p_l(x,\VFS)(1- p_l(x,\VFS)) \leq \frac14$ for all $x,\VFS,l$, and so the diagonal entries are positive.
The off-diagonal entries are nonnegative since $p_l(x,\VFS) \geq 0$ for all $x,\VFS,l$, and the rows sum to one since $\sum_l p_l (x,\VFS) = 1$ for all $x, \VFS$. When the $q_i$'s are inseparable, under Assumption~\ref{asm: inseparable}, $\bar P(x)$ is irreducible for sufficiently large $N$.
Moreover, $\bar P(y)$ is aperiodic since its diagonal is positive.

Now let 
\begin{equation*}
  \bar H:= \int_{s=0}^1 \nabla^2 l_N(s \bar f + (1-s) f) \, ds = \bar Q-I,
\end{equation*}
where
\begin{equation*}
 \bar Q:= \int_{s=0}^1 \bar P( s \bar f + (1-s) f) \, ds.
\end{equation*}
Note that $\bar Q$ is stochastic, and that it inherits symmetry, irreducibility, and aperiodicity from $\bar P$.
It follows from symmetry and irreducibility that $L^{-1}\1$ is the unique invariant distribution of $\bar Q$, and therefore by Theorem~5.5 in~\onlinecite{meyer1975role}
the group inverse $\bar H^\#$  of $\bar H$ exists and is given by
\begin{equation}
  \label{eqn: formula for group inverse of bar beta}
  \bar H^\# 
  = L^{-2} \1\1^t + ( \bar H-  L^{-2}\1 \1^t)^{-1}.
\end{equation}
The group inverse is a particular type of generalized inverse similar to the Moore-Penrose inverse. (In fact, the group inverse of $H$ is the same as its Moore-Penrose inverse. We prefer to call it the group inverse since we will use the  results of Meyer~\cite{meyer1975role} on group inverses of generator matrices.) The group inverse is characterized by the properties
\begin{equation*}
  \bar H\bar H^.\# = \bar H^\# \bar H, \text{ } \bar H^\# \bar H\bar H^\#  =\bar H^\#, \text{ and } \bar H\bar H^\# \bar H=\bar H.
\end{equation*}
We refer to \onlinecite{meyer1975role,golub1986using} for details.
For future reference, we also note that by Theorem 2.2 in \onlinecite{meyer1975role}, we have
\begin{equation}
  \label{eqn: formula for beta sharp time beta}
  \bar H^\# \bar H
  = I - L^{-2} \1 \1^t.
\end{equation}

We will use $\bar H^\#$ as a building block of our generalized inverse $\bar \Gamma$.
We have
\begin{equation*}
  \bar B = \int_{s=0}^1 \bar \RFN'(s \bar v + (1-s) v) \, ds = 
  \begin{pmatrix}
    \bar H& 0 \\
    \dw & I 
  \end{pmatrix},
\end{equation*}
where
\begin{equation*}
  \dw := \int_{s=0}^1 \frac{\partial \omega}{\partial y}(s \bar v + (1-s) v) \, ds.
\end{equation*}
We define
\begin{equation*}
  \bar \Gamma =
  \begin{pmatrix}
    \bar H^\# & 0 \\
    - \dw \bar H^\# & I
  \end{pmatrix}.
\end{equation*}
Observe that by~\eqref{eqn: formula for beta sharp time beta}, 
\begin{equation*}
  \bar \Gamma \bar B =
  \begin{pmatrix}
    \bar H^\# & 0 \\
    - \dw \bar H^\# & I
  \end{pmatrix}
  \begin{pmatrix}
    \bar H& 0 \\
    \dw & I 
  \end{pmatrix}
  =
  \begin{pmatrix}
    I - L^{-2} \1 \1^t & 0 \\
    \dw (I - L^{-2} \1 \1^t) & I 
  \end{pmatrix}.
\end{equation*}
Therefore, since
\begin{equation*}
  \1^t \bar f = \1^t f = 0
\end{equation*}
by the constraint~\eqref{eqn: constraint}, we have
\begin{equation}
  \label{eqn: left inverse property}
\bar \Gamma \bar B (\bar v - v) = \bar v - v.
\end{equation}

Equation~\eqref{eqn: left inverse property} is one of the desired properties of $\bar \Gamma$. We will now show in addition that
\begin{equation}
  \label{eqn: convergence of bar b prime}
  \bar \Gamma \xrightarrow{as}
  \Gamma :=
  \begin{pmatrix}
    H(f)^\# & 0 \\
    - dw(f,w) H(f)^\# & I
  \end{pmatrix}.
\end{equation}
First, we show that
\begin{equation}
  \label{eqn: convergence of bar H}
  \bar H \xrightarrow{as} H(f).
\end{equation}
We have
\begin{align*}
  \bar H_{lm} &=  h_1 + h_2, \\
\end{align*}
where
\begin{align*}
  h_1 &:= \int_{s=0}^1 \frac{\partial l_N}{\partial y_l \partial y_m}(s \bar f + (1-s) f) - \frac{\partial l_N}{\partial y_l \partial y_m}(f) \, ds, \text{ and } \\
  h_2 &:= \frac{\partial l_N}{\partial y_l \partial y_m}(f).
\end{align*}
Using the uniform bounds on third derivatives of the likelihood~\eqref{eqn: bound on third derivatives},
\begin{align*}
  \lvert h_1 \rvert & \leq 2 L \max_k \lvert \bar f_k - f_k \rvert, 
\end{align*}
and therefore
\begin{equation*}
h_1 \xrightarrow{as} 0,
\end{equation*}
since as discussed above we have $\bar f \xrightarrow{as} f$ by the results of Geyer~\cite{geyer1994estimating}.
By~\eqref{eqn: limits of derivatives of likelihood},
\begin{equation*}
  h_2 \xrightarrow{as} H(f).
\end{equation*}
Thus, \eqref{eqn: convergence of bar H} holds.
A similar argument using~\eqref{eqn: bound on second derivatives of w} and~\eqref{eqn: limits of bar w} shows that
\begin{equation*}
  \bar dw \xrightarrow{as} dw(f, \omega).
\end{equation*}

To complete the proof that $\bar \Gamma \xrightarrow{as} \Gamma$, consider formula~\eqref{eqn: formula for group inverse of bar beta} for $\bar H^\#$. The right hand side
\begin{equation*}
  L^{-2} \1\1^t + ( \bar H-  L^{-2}\1 \1^t)^{-1}
\end{equation*}
of~\eqref{eqn: formula for group inverse of bar beta} is continuous as a function of $\bar H$ over an open neighborhood of any $\bar H$ so that the inverse exists, since the matrix inverse is a continuous function on a neighborhood of any nonsingular matrix. In particular, it is continuous at $H(f)$, since $H(f)$ is stochastic, symmetric, irreducible, and aperiodic by the same argument that shows that $\bar H$ has these properties. It follows that
\begin{equation*}
  \bar H^\# \xrightarrow{as}  L^{-2} \1\1^t + ( H(f)-  L^{-2}\1 \1^t)^{-1}= H(f)^\#,
\end{equation*}
hence
\begin{equation*}
  \bar \Gamma \xrightarrow{as} \Gamma.
\end{equation*}
Combining the results of the previous sections yields the following theorem.
\begin{theorem}
  Under Assumptions~\ref{asm: inseparable},~\ref{asm: boundedness},~\ref{asm: ergodicity}, and~\ref{asm: clt}, we have
  \begin{equation}
  \sqrt{N} (\bar v- v) \xrightarrow{d} N(0,\Gamma A(v) \Gamma^t).
  \label{eq:statement_of_clt}
\end{equation}
\end{theorem}
\begin{proof}
By the results of Section~\ref{sec: linearization},
\begin{equation*}
 - \bar \Gamma \sqrt{N} (\bar \RFN(v) - \RFN(v)) = \sqrt{N} (\bar v- v), 
\end{equation*}
and the result follows from Slutsky's theorem and the CLT for $\bar \RFN(v)$. 
\end{proof}

\subsection{A Simple Expression for $\Gamma$}
The matrix  $\Gamma$ can be further simplified by observing that
\begin{align*}
    H_{lm}(f)  
        =& - \delta_{lm} \E [ p_l(X, f)] + \E [p_l(X, f) p_m(X, f)] 
            \\
        =& - \delta_{lm} \sum_{i=1}^L \exp( f_i)  \int p_l(x, f) q_i(x) dx
            + \sum_{i=1}^L \exp( f_i )  \int p_l(x, f) p_m(x, f) q_i(x) dx
            \\
        =& - \delta_{lm} \sum_{i=1}^L \exp(f_i)  \int \frac{q_l(x)\exp(f_l)}{\sum_j q_j(x)\exp (f_j) } q_i(x) dx
            + \sum_{i=1}^L \exp(f_i)  \int \frac{q_l(x) \exp(f_l ) }{\left(\sum_k q_k(x) \exp( f_k )\right)} p_m (x, f) q_i(x) dx
            \\
        =& -    \int \left(\frac{ \delta_{lm} q_l(x)\exp(f_l)}{\sum_j q_j(x)\exp (f_j) }
            -  \frac{q_l(x) \exp(f_l ) }{ \sum_k q_k(x) \exp( f_k )} p_m(x, f) \right) \left( \sum_{i=1}^L  q_i(x)  \exp(f_i)  \right) dx
            \\
        =& -    \int \left(  \delta_{lm} q_l(x)\exp(f_l)
            - p_m(x, f) q_l(x) \exp(f_l )  \right) dx
            \\
        =& -  \delta_{lm} \int q_l(x) \exp(f_l) dx 
            + \int   p_m(x, f) q_l(x)\exp(f_l) dx
            \\
        =& -  \delta_{lm} \kappa_l 
            + \frac{\kappa_l}{c_l}\int p_m(x, f) q_l(x) dx,
\end{align*}
where the final line follows from the fact that  $ q_l(x)/ c_l $ is a probability density and consequently integrates to one.
Similarly, we have that
\begin{align*}
    dw_{il}(v)
    &= \sum_{i=1}^L \exp( f_i)  \int \frac{w_i(X) p_l(X,f)}{\qmix(X,f)}  q_i(x)dx \\
    &=  \int \left( \frac{w_i(X) }{\sum_{j=1}^L q_j(x)\exp (f_j)} \right) \left( \frac{q_l(x) \exp (f_l) }{\sum_{j=1}^L q_j(x)\exp (f_j)} \right)\left( \sum_{i=1}^L  \exp( f_i) q_i(x) \right) dx \\
    &=  \int \left( \frac{w_i(X) }{\sum_{j=1}^L q_j(x)\exp (f_j)} \right) q_l(x) \exp (f_l)  dx 
    \\
    &=  \int \left( \frac{w_i(X) }{\sum_{j=1}^L q_j(x)\exp (f_j)} \right) q_l(x) \exp (f_l)  dx
\end{align*}


\section{Data-driven Estimates of the Asymptotic Variance}

Here, we cover in more detail the derivation of our expressions for the asymptotic variance of MBAR estimates from trajectory autocorrelations.
We consider two cases: the case where each state is sampled independently of all the other states with its own Markov chain (as in many Umbrella sampling and Alchemical Simulation calculations),
and the case where all states are sampled together using a single joint Markov chain (as in Parallel Tempering and Hamiltonian Replica Exchange).

\subsection{Independent Sampling of States}
As observed in the main text, if states are sampled independently then $\Xi^{lm}$ is nonzero if and only if $l=m$.
We can therefore write 
\begin{align}
    \big(\grad \Phi^T(v)  &{\Gamma} A \Gamma^T \grad \Phi(v) \big)_{im} 
         = \sum_{j=1}^{L+M} \sum_{k}^{L+M} \sum_{l=1}^L  (\grad \Phi^T {\Gamma}(v))_{ij} \kappa_l  \Xi_{jk}^{ll} (\Gamma^T \grad \Phi(v))_{km} \nonumber \\
        =&  \sum_{j=1}^{L+M} \sum_{k}^{L+M} \sum_{l=1}^L  (\grad \Phi^T {\Gamma}(v))_{ij} 
                \kappa_l  
                \bigg( 
                    \cov \left\{ \xi_j(X_t^l,v), \xi_k(X_t^l,v) \right\} +
                    2\sum_{\tau=1}^\infty \cov\left\{\xi_j(X_t^l,v),\xi_k(X_{t+\tau}^l,v)\right\}
                \bigg)
                (\Gamma^T \grad \Phi(v))_{km} \nonumber \\
        =& \sum_{l=1}^L  \kappa_l    \bigg( 
                    \cov \left\{ \sum_{j=1}^{L+M} (\grad \Phi^T {\Gamma}(v))_{ij} 
\xi_j(X_t^l,v), \sum_{k}^{L+M} (\Gamma^T \grad \Phi(v))_{km}\xi_k(X_t^l,v) \right\}  \nonumber \\
         & \qquad \quad + 2 \sum_{\tau=1}^\infty \cov\left\{ \sum_{j=1}^{L+M} (\grad \Phi^T {\Gamma}(v))_{ij} 
   \xi_j(X_t^l,v), \sum_{k}^{L+M}(\Gamma^T \grad \Phi(v))_{km} \xi_k(X_{t+\tau}^l,v)\right\}
                \bigg)
                 \nonumber \\
        =& 
        \sum_{l=1}^L 
            \bigg( 
                \cov \left\{ \chi_i (X_t^l), \chi_m (X_t^l) \right\} +
                2\sum_{\tau=1}^\infty \cov\left\{\chi_i(X_t^l),\chi_m(X_{t+\tau}^l)\right\}
            \bigg)
\end{align}
which recovers equation~52 in the main text.

\subsection{Joint Sampling of States}
Alternatively, we can consider the case where, for a given $t$, the $X_i^t$'s are collected jointly for all $i$ from a single Markov chain.
By definition we then have the same number of points in every state and $\kappa_i = (1/L)$  for all  $i$.
We then write
\begin{align}
    \big(\grad \Phi^T(v)  &{\Gamma} A \Gamma^T \grad \Phi(v) \big)_{in} 
    = \sum_{j=1}^{L+M} \sum_{k}^{L+M} \sum_{l=1}^L \sum_{m=1}^L    (\grad \Phi^T {\Gamma}(v))_{ij} (1/L) \Xi_{jk}^{lm} (\Gamma^T \grad \Phi(v))_{kn} \nonumber \\
    =&  \frac{1}{L} \sum_{j=1}^{L+M} \sum_{k}^{L+M} \sum_{l=1}^L  \sum_{m=1}^L (\grad \Phi^T {\Gamma}(v))_{ij} 
                \kappa_l  
                \bigg( 
                    \cov \left\{ \xi_j(X_t^l,v), \xi_k(X_t^m,v) \right\} + \nonumber \\
                    &\qquad \quad + 2\sum_{\tau=1}^\infty \cov\left\{\xi_j(X_t^l,v),\xi_k(X_{t+\tau}^m,v)\right\}
                \bigg)
                (\Gamma^T \grad \Phi(v))_{kn} \nonumber \\
    =& (1 / L) \bigg( 
                    \cov \left\{ \sum_{l=1}^L  \sum_{j=1}^{L+M} (\grad \Phi^T {\Gamma}(v))_{ij} 
\xi_j(X_t^l,v), \sum_{m=1}^L  \sum_{k}^{L+M} (\Gamma^T \grad \Phi(v))_{km}\xi_k(X_t^m,v) \right\}  \nonumber \\
         & \qquad \quad + 2 \sum_{\tau=1}^\infty \cov\left\{ \sum_{l=1}^L  \sum_{j=1}^{L+M} (\grad \Phi^T {\Gamma} (v))_{ij} 
   \xi_j(X_t^l,v), \sum_{m=1}^L  \sum_{k}^{L+M}(\Gamma^T \grad \Phi(v))_{kn} \xi_k(X_{t+\tau}^m,v)\right\}
                \bigg)
                 \nonumber \\
        =& 
            \bigg( 
                \cov \left\{ \sum_{l=1}^L \chi_i (X_t^l)), \sum_{m=1}^L \chi_n (X_t^m) \right\} +
                2\sum_{\tau=1}^\infty \cov\left\{\sum_{l=1}^L  \chi_i(X_t^l), \sum_{m=1}^L \chi_n(X_{t+\tau}^m)\right\}
            \bigg)
\end{align}

\end{document}